\documentclass[11pt,a4paper]{article}

\usepackage[english]{babel}
\usepackage[utf8]{inputenc}
\usepackage[T1]{fontenc}
\usepackage{natbib}
\usepackage{amsmath}
\usepackage{amssymb}
\usepackage{amsthm}
\usepackage{hyperref}
\usepackage{graphicx}
\usepackage{times}
\usepackage{bm}
\usepackage{multicol}
\usepackage{bbold}
\usepackage[ruled,noend]{algorithm2e}
\makeatletter
\renewcommand{\algocf@captiontext}[2]{#1\algocf@typo. \AlCapFnt{}#2}

\def\@algocf@capt@plain{top}
\renewcommand{\algocf@makecaption}[2]{%
  \addtolength{\hsize}{\algomargin}%
  \sbox\@tempboxa{\algocf@captiontext{#1}{#2}}%
  \ifdim\wd\@tempboxa >\hsize%
    \hskip .5\algomargin%
    \parbox[t]{\hsize}{\algocf@captiontext{#1}{#2}}%
  \else%
    \global\@minipagefalse%
    \hbox to\hsize{\box\@tempboxa}%
  \fi%
  \addtolength{\hsize}{-\algomargin}%
}
\makeatother

\usepackage{tikz}
\usetikzlibrary{matrix,fit,backgrounds,positioning}
\usepackage{pgfplots}

\usepackage[margin=2.5cm]{geometry}

\usepackage{setspace}
\usepackage{natbib}
\def\Clust{K}
\def\dataspace{\mathcal{D}}
\def\dglob{d_{\mathrm{\scriptscriptstyle glob}}}
\def\dloc{d_{\mathrm{\scriptscriptstyle loc}}}

\def\globparam{\beta}
\def\iperm{l}

\def\locparam{\mu}

\def\Npermstrat#1{L_{#1}}

\def\obsdata{\mathbf{y}}
\def\obsdataUn{y}
\def\param{\boldsymbol{\theta}}
\def\parprop{\boldsymbol{\theta}}
\def\parset{\Theta}
\def\proj{\varphi_{\obsdata}}

\def\simdata{\mathbf{z}}

\def\eps{\varepsilon}
\def\projsharp{{\varphi_{\mathbf{y}}}_{\sharp}}

\DeclareMathOperator*{\argmin}{arg\,min}

\graphicspath{{./art/}}

\theoremstyle{plain}
\newtheorem{theorem}{Theorem}[section]
\newtheorem{lemma}[theorem]{Lemma}
\newtheorem{corollary}[theorem]{Corollary}

\theoremstyle{definition}

\theoremstyle{remark}

\begin{document}

\title{Permutations accelerate Approximate Bayesian Computation}

\author{
Antoine Luciano\\
CEREMADE, Université Paris Dauphine--PSL, Paris, France\\
\texttt{luciano@ceremade.dauphine.fr}
\and
Charly Andral\\
CEREMADE, Université Paris Dauphine--PSL\\
\texttt{andral@ceremade.dauphine.fr}
\and
Christian P. Robert\\
CEREMADE, Université Paris Dauphine--PSL\\
\texttt{xian@ceremade.dauphine.fr}
\and
Robin J. Ryder\\
Department of Mathematics, Imperial College London, U.K.\\
\texttt{r.ryder@imperial.ac.uk}
}


\maketitle

\begin{abstract}
Approximate Bayesian Computation (ABC) methods have become essential tools for performing inference when likelihood functions are intractable or computationally prohibitive. However, their scalability remains a major challenge in hierarchical or high-dimensional models. In this paper, we introduce \emph{permABC}, a new ABC framework designed for settings with both global and local parameters, where observations are grouped into exchangeable compartments.
Building upon the Sequential Monte Carlo ABC (ABC-SMC) framework, permABC exploits the exchangeability of compartments through permutation-based matching, significantly improving computational efficiency. 

We then develop two further, complementary sequential strategies: \emph{Over Sampling}, which  facilitates early-stage acceptance by temporarily increasing the number of simulated compartments, and \emph{Under Matching}, which relaxes the acceptance condition by matching only subsets of the data.

These techniques allow for robust and scalable inference even in high-dimensional regimes. Through synthetic and real-world experiments — including a hierarchical Susceptible-Infectious-Recover model of the early COVID-19 epidemic across 94 French departments — we demonstrate the practical gains in accuracy and efficiency achieved by our approach.
\end{abstract}


\section{Introduction}

Approximate Bayesian Computation (ABC) \citep{marinApproximateBayesianComputational2012a} has emerged as a pivotal tool for statistical inference in scenarios where the likelihood function is computationally infeasible or explicitly unavailable. Initially applied in population genetics by \citet{tavareInferringCoalescenceTimes1997} and extended by \citet{pritchardPopulationGrowthHuman1999a}, ABC has found applications across various fields such as climate science \citep{holden2018abc}, ecology \citep{piccioni2022calibration}, and epidemiology \citep{sissonHandbookApproximateBayesian2018}. The key idea behind ABC is to simulate data under the prior predictive distribution and retain simulations $\simdata$ that are close to the observed data $\obsdata$, thereby avoiding explicit likelihood computations. This comparison relies on a metric $d$ and a tolerance threshold $\varepsilon > 0$: we accept when $d(\obsdata, \simdata)\leq\varepsilon$. ABC methods are thus particularly useful for highly complex models, for which it is difficult to calculate the likelihood but easy to simulate synthetic data.

Challenges arise for large datasets and for large parameter space. 
To reduce the dimensionality of the data, summary statistics are typically employed. Practitioners aim to select informative summaries and to choose $\varepsilon$ as small as computationally feasible, as tighter thresholds yield posterior approximations of higher fidelity. The problem of constructing such summaries has received considerable attention. For example, \citet{fearnheadConstructingSummaryStatistics2012} proposed an inferential criterion for summary statistic selection to improve posterior accuracy. Many applications rely on a small number of expert-informed summary statistics.

Significant methodological progress has been made by integrating Monte Carlo techniques into ABC, to facilitate the exploration of the parameter space. Among these, Markov Chain Monte Carlo (MCMC) and Sequential Monte Carlo (SMC) strategies have been particularly influential. \citet{marjoramMarkovChainMonte2003a} proposed ABC-MCMC, which uses a Markov chain to guide parameter proposals toward regions of high posterior mass. Later, \citet{sisson2007sequential} introduced the Partial Rejection Control (PRC-ABC) algorithm, adopting a particle-based perspective. However, this method suffered from biased weighting issues \citep{sissonCorrectionSissonSequential2009}. 
To resolve these limitations, \citet{beaumontAdaptiveApproximateBayesian2009b} proposed ABC-PMC, which combines importance sampling and adaptive kernel scaling. This innovation led to improved efficiency and more accurate posterior approximations. Subsequently, \citet{delmoralAdaptiveSequentialMonte2012} provided a theoretical framework for ABC-SMC using backward kernels and adaptive thresholding, which further enhanced robustness and convergence. These algorithmic refinements, particularly in the context of SMC, have made ABC applicable to increasingly complex models and larger datasets. Still, the field continues to evolve, as illustrated by recent work on guided proposals and adaptivity \citep{picchini2024guided}.

Despite these advances, standard ABC approaches remain challenged by high-dimensional problems and parameter dependencies. To address this, Gibbs-inspired strategies have been proposed, which capitalize on model structure to update parameters conditionally. \citet{clarte2021componentwise} and \citet{rodrigues2020likelihood} introduced independently a likelihood-free Gibbs sampling scheme where each parameter block is updated using its approximate conditional distribution via ABC. 
\citet{clarte2021componentwise} established theoretical convergence guarantees for ABC-Gibbs under partial independence, demonstrating substantial gains in efficiency and approximation quality, particularly in hierarchical contexts. These contributions expand the ABC toolbox toward more scalable and modular inference schemes.

Motivated by these developments, we turn our attention to a specific class of hierarchical models commonly encountered in applied settings where observations are naturally grouped into independent subpopulations or \emph{compartments}. 
As a motivating example, to which we return in Section \ref{sec:sir}, consider the popular Susceptible-Infectious-Recovered model in epidemiology: we wish to apply this model to the early phases of the COVID epidemic in France, with data observed locally at the level of each of $94$ \emph{d\'epartements}; initial exploration shows that certain parameters are global (constant across the whole country) but others are local (they vary between the départements).   
This structure suggests a new inference framework that exploits exchangeability across compartments.

More generally, we consider a setting where the data are divided into $K$ compartments. In each compartment $k = 1, \dots, K$, we observe $n$ data points $\obsdata_k = (\obsdataUn_k^1, \dots, \obsdataUn_k^n) \in \dataspace$, which depend on a shared global parameter $\globparam$ (with prior $\pi_\globparam$) and a compartment-specific local parameter $\locparam_k$. We assume that $(\mu_k, \obsdata_k)_{k=1}^K$ are exchangeable conditional on $\globparam$; a typical application is when the local parameters are independent and identically distributed $\mu_1,\ldots,\mu_K\sim\pi_\mu$. Collectively, the full dataset is denoted by $\obsdata := (\obsdata_1, \dots, \obsdata_K) \in \dataspace^K$, and the parameter vector by $\param = (\globparam, \locparam_{1}, \dots, \locparam_{K}) \in \parset$, with total dimension $d = K\dloc + \dglob$, where $\dloc$ and $\dglob$ are respectively the dimensions of the local and global parameters. 
Such a hierarchical model is illustrated in Figure \ref{fig:hier}.

Let $\mathcal S_K$ be the symmetric group over $\{1,\ldots,K\}$. 
We define the permuted parameter vector $\param_{\sigma} = (\globparam, \locparam_{\sigma(1)}, \dots, \locparam_{\sigma(K)})$, where $\sigma \in \mathcal{S}_K$ is a permutation of the compartments. We assume that each $\obsdata_k$ is generated independently from a common likelihood function $g$, conditional on its corresponding parameters. In this context, exchangeability implies that the joint data distribution is invariant under permutations: $f(\obsdata\mid \param) = \prod_{k = 1}^K g(\obsdata_k\mid \globparam,\locparam_k) = f(\obsdata_{\sigma} \mid \param_{\sigma})$ for any $\sigma \in \mathcal{S}_K$. 

We further assume that the likelihood $g$ is intractable or computationally costly, though we can simulate from it. That is, given $(\globparam,\locparam_k)$, we can generate pseudo-data $\obsdata_k$. By abuse of notation, we refer to the global simulator as the global likelihood $f$ throughout.


Building on the exchangeability property, the main innovation of our work is to simulate synthetic data, then apply a well-chosen permutation to the compartments before checking the ABC acceptance criterion, yielding our \emph{permABC} algorithm, introduced in Section~\ref{secabc_to_permabc}. This strategy increases the acceptance probability by several orders of magnitude. 
We consider sequential improvements in Section~\ref{sec:sequential}: we first introduce permABC-SMC, then  define new sequences of intermediate distribution for sequential strategies (over-sampling and under-matching); they are made possible by our use of permutations, and they improve efficiency and robustness in challenging regimes such as outliers or model mis-specification. Section~\ref{sec:numeric} provides a thorough empirical assessment of the proposed methods, including comparisons with classical ABC techniques, synthetic experiments highlighting the gains of permutation-based inference, and a real-world application to COVID-19 epidemic data in France.

\begin{figure}[ht]
\begin{center}
\begin{tikzpicture}
\matrix (mat) [matrix of math nodes, column sep=10pt, row sep=0pt, name=mat] {
     & & \pi_\mu (\cdot) & & \\[23pt]
     \mu_1 & \mu_2 & \ldots &   & \mu_K \\[20pt] 
     \mathbf{y}_1 & \mathbf{y}_2 & \ldots &  & \mathbf{y}_K \\[20pt] 
     & & \beta  & & \\[20pt]
     & & \pi_\beta (\cdot) & & \\
};

\foreach \column in {1, 2, 5}
{   
  \draw[->,>=latex] ([yshift=-2pt]mat-1-3) -- ([yshift=2pt]mat-2-\column.north);
  \draw[->,>=latex] ([yshift=-2pt]mat-2-\column) -- ([yshift=2pt]mat-3-\column);
  \draw[<-,>=latex] ([yshift=-2pt]mat-3-\column) -- ([yshift=2pt]mat-4-3);
}

\draw[<-,>=latex] ([yshift=-2pt]mat-4-3.south) -- ([yshift=2pt]mat-5-3.north);

\node[align=center, right=6em of mat-3-3, yshift=1.7em, inner sep=5pt] (exchangeable) 
    {\small $(\mu_k, \mathbf{y}_k)_{k=1}^K$ are exchangeable\\ \small conditionally on $\beta$};

\end{tikzpicture}
\end{center}
\caption{Hierarchical model framework with exchangeable priors: $\mu_k \sim \pi_\mu(\cdot)$, $\beta \sim \pi_\beta(\cdot)$.}
\label{fig:hier}
\end{figure}
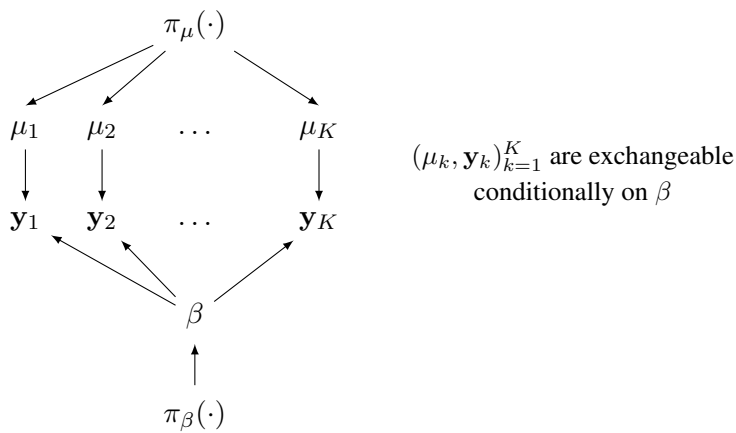

\section{From ABC to permABC}
\label{secabc_to_permabc}

\subsection{Introduction to ABC}

The standard rejection ABC algorithm, often referred to as \emph{Vanilla ABC} (see Algorithm~\ref{alg:abc_vanilla}), proceeds by drawing parameters from the prior distribution, $\param \sim \pi(\cdot)$, and then simulating synthetic data from the likelihood, $\simdata \sim f(\cdot \mid \param)$. The pair $(\param, \simdata)$ is accepted if the summary statistics of the simulated data are sufficiently close to those of the observed data, that is, if $d\{\eta(\obsdata), \eta(\simdata)\} \leq \varepsilon$ for some tolerance level $\varepsilon > 0$ and some summary statistic $\eta$. This procedure yields samples from the approximate joint distribution $\pi_{\varepsilon}\{\param,\simdata \mid \eta(\obsdata)\}$; we marginalize out $\simdata$ to get the ABC posterior

\begin{equation}
    \label{eqabc}
    \begin{split}
        \pi_{\varepsilon}\{\param \mid \eta(\obsdata)\} &\propto
        \int_{\dataspace^\Clust} \pi_{\varepsilon}\{\param, \simdata \mid \eta(\obsdata)\} \, d\simdata \\\\
        &\propto \pi(\param) \int_{A_{\obsdata,\varepsilon}} f(\simdata \mid \param) \, d\simdata,
    \end{split}
\end{equation}
where $A_{\obsdata,\varepsilon} = \{\simdata \in \dataspace^\Clust \mid d\{\eta(\obsdata), \eta(\simdata)\} \leq \varepsilon\}$ denotes the acceptance region.

Combining an informative summary statistic $\eta$ with a small tolerance $\varepsilon$ typically leads to a good approximation of the true posterior distribution, so that $\pi_{\varepsilon}\{\param \mid \eta(\obsdata)\} \approx \pi(\param \mid \obsdata)$. This approximate posterior—often called a \emph{pseudo-posterior}—is obtained by replacing the true likelihood $f(\obsdata \mid \param)$ with the integral of the likelihood over the neighborhood $A_{\obsdata,\varepsilon}$ of the observed data, effectively using it as a surrogate for the intractable likelihood.

Throughout this work, we assume for clarity that the data are compared directly, i.e., $\eta = \mathrm{Id}$. Under this assumption, the acceptance region simplifies to the $\varepsilon$-ball $B_{\varepsilon}(\obsdata)$ centered at $\obsdata$ with respect to the distance $d$. Nonetheless, the proposed method naturally extends to the use of low-dimensional summary statistics.

In our numerical examples, we take $\dataspace = \mathbb R^n$ and use the distance $d(\cdot,\cdot)$ between two datasets in \( \dataspace^K \), each consisting of \( K \times n \) observations, defined as follows: 
\begin{equation}
    \label{eqdistance}
    d(\obsdata, \simdata)^2 = {\sum_{k=1}^K w_k^2 \| \obsdata_k - \simdata_k \|_2^2},
\end{equation}
where \( \| \cdot \|_2 \) denotes the Euclidean norm, and the weights \( w_k \) are used to control the relative influence of each compartment in the distance computation. The weight vector \( w_{1:K} \) can either be manually specified by the user or selected adaptively to prevent any single compartment from dominating the distance, as suggested by \citet{prangle2017adapting}.

This naive ABC algorithm faces several limitations, including poor scalability in the parameter space and an exponentially increasing computational cost due to the declining acceptance rate as \( \varepsilon \) decreases. In the context of the hierarchical model of Figure \ref{fig:hier}, the acceptance rate further diminishes as the number of compartments \( K \) increases. Simulating a dataset \( \simdata \) from the prior predictive distribution that matches all compartments of \( \obsdata \) within a small tolerance \( \varepsilon \) can be particularly challenging. 

\begin{algorithm}[ht]
\caption{Vanilla Approximate Bayesian Computation}
\label{alg:abc_vanilla}
\textbf{Input: }observed dataset $\obsdata$, number of iterations $N$ , threshold $\eps>0$, summary statistic $\eta$.\\
\textbf{Output: }a sample $(\param^{(1)}, \dots, \param^{(N)})$ from $\pi_\eps\{\cdot\mid\eta(\obsdata)\}$.\\

\BlankLine

\For{$i = 1,\dots, N$}{
    \Repeat{$d\{\eta(\simdata^{(i)}),\eta(\obsdata)\})\leq \eps$}{
        $\param^{(i)} \sim \pi(\cdot)$\;
        $\simdata^{(i)}\sim f(\cdot\mid\param^{(i)})$\;
        }
    }
\end{algorithm}

\subsection{Motivations of permABC}

A common approach in ABC to increase the acceptance rate is to exploit the exchangeability of the compartments by using a permutation-invariant summary statistic, such as the sorted vector. For instance, when we have a one-dimensional observation (or statistic) per compartment, $\obsdata \in \mathbb{R}^K$, it is natural to use the ordering $\eta(\obsdata) = \text{sort}(\obsdata)$ as a summary. This reduces the effective complexity from $K!$ permutations to a canonical representation, thereby simplifying the comparison between datasets.

Moreover, in the univariate case, sorting is known to minimize the distance over all permutations of simulated data. Specifically, $d\{\text{sort}(\obsdata), \text{sort}(\simdata)\} = \min_{\sigma \in \mathcal{S}_K} d(\obsdata, \simdata_{\sigma})$, where $\mathcal{S}_K$ denotes the set of all permutations of $\{1,\dots,K\}$. This property underlies several recent developments in ABC, including those based on optimal transport and Wasserstein distances \citep{bernton2019approximate}, where sorting plays a central role in defining distances between empirical distributions.

However, this elegant result no longer holds in the multivariate case. When each $\obsdata_k$ is a vector in $\mathbb{R}^n$ with $n > 1$, there is no natural way to define a deterministic and permutation-optimal ordering. As a result, we cannot rely on sorting to minimize $d(\obsdata, \simdata_{\sigma})$ in general, and alternative strategies are needed as soon as the dimension is at least 2. Instead, we are compelled to solve the optimization problem directly, which leads us to the following permABC acceptance criterion:

\begin{equation}
    \label{eqopti}
    \min_{\sigma \in \mathcal{S}_K} d(\obsdata, \simdata_{\sigma}) \leq \varepsilon.
\end{equation}

This condition can be interpreted as accepting if our simulated dataset $\simdata$ matches, up to a permutation, the observed dataset $\obsdata$, and it effectively enhances the acceptance rate by a factor of $K!$ when $\varepsilon$ is small. This optimization task can be formulated as a two-dimensional linear assignment problem, or equivalently a bipartite matching problem, and is efficiently solved using a modified Jonker–Volgenant algorithm \citep{jonker1987shortest}, with a computational complexity of $\mathcal{O}(K^3)$ \citep{crouseImplementing2DRectangular2016a}.
While this reduces the combinatorial cost from $K!$ to polynomial time, the $\mathcal{O}(K^3)$ complexity remains substantial. Nevertheless, it makes problems with up to $K = 100$ compartments tractable in practice (see the real-world application in Section~\ref{sec:sir}).

However, this approach does not yield samples from the pseudo-posterior defined in \eqref{eqabc}, but rather from the following permutation-invariant version:

\begin{equation} 
    \label{eqpermabc_tilde}
    \begin{split}
        \Tilde \pi_{\varepsilon} (\param \mid \obsdata)  & \propto \int_{\dataspace^K} \Tilde \pi_{\varepsilon} (\param , \simdata \mid \obsdata) \, d\simdata \\
        & \propto \pi(\param) \int_{\Tilde A_{\obsdata,\varepsilon}} f(\simdata \mid \param) \, d\simdata,
    \end{split}
\end{equation}
where the acceptance region is defined as
\[
\tilde{A}_{\obsdata,\varepsilon} = \left\{\simdata \in \dataspace^K \,\middle|\, \min_{\sigma \in \mathcal{S}_K} d(\obsdata, \simdata_{\sigma}) \leq \varepsilon \right\} = \bigcup_{\sigma \in \mathcal{S}_K} B_{\varepsilon}(\obsdata_{\sigma}),
\]
with \( B_{\varepsilon}(\obsdata_{\sigma}) \) denoting the \( \varepsilon \)-ball centered at \( \obsdata_{\sigma} \) for the distance \( d \).

In this formulation, the likelihood \( f(\obsdata \mid \param) \) is approximated by its integral over a neighborhood of all permutations of the observed dataset:
\[
\int_{\tilde{A}_{\obsdata,\varepsilon}} f(\simdata \mid \param) \, d\simdata \approx \sum_{\sigma \in \mathcal{S}_K} f(\obsdata_{\sigma} \mid \param).
\]

Inference in this setting targets the global parameter \( \globparam \), but is only defined up to a permutation of the local components \( \locparam_{1:K} \), effectively yielding posterior samples of the form \( \Tilde{\param} = (\globparam, \{\locparam_1, \dots, \locparam_K\})\),
where the local parameters are treated as an unordered set. This formulation is particularly suitable when the local parameters are considered nuisance variables, as in many hierarchical modeling contexts. We discuss below how to recover a pseudo-posterior of each $\mu_k$ if needed.

To our knowledge, this permutation-based version of ABC has not been formally studied in the literature. We therefore introduce a novel methodology, termed \emph{permABC}, which builds upon this idea while addressing its limitations. In particular, our method enables proper inference over the full parameter vector \( \param \), including the ordering of local components, even in regimes where \( K \) is large, and does so while preserving a high acceptance rate.

\subsection{The permABC algorithm}
\label{secpermabc_star}

Restoring identifiability among compartments is essential for enabling inference on the full parameter vector. One natural way to achieve this is by applying a suitable permutation to the accepted parameters, such that the simulated data are optimally aligned with the observed dataset. However, explicitly testing all \( K! \) permutations is computationally intractable.

Fortunately, whenever the acceptance condition \eqref{eqopti} is satisfied, there exists at least one permutation of the simulated dataset that meets the ABC criterion, namely the optimal one minimizing the distance. This permutation, denoted by \( \sigma^* = \argmin_{\sigma \in \mathcal{S}_K} d(\obsdata, \simdata_{\sigma}) \), provides a natural way to reorder the parameter vector.

In practice, the \emph{permABC} algorithm proceeds similarly to standard ABC: we sample \( \parprop \sim \pi(\cdot) \) and simulate data \( \simdata \sim f(\cdot \mid \parprop) \). If the minimum distance over all permutations satisfies the tolerance condition, we accept and return the pair \( (\parprop_{\sigma^*}, \simdata_{\sigma^*}) \). This induces a new pseudo-posterior distribution, denoted \( \pi_{\varepsilon}^*(\cdot, \cdot \mid \obsdata) \), which corresponds to the pushforward of the distribution \( \tilde{\pi}_\varepsilon \) defined in \eqref{eqpermabc_tilde} under the optimal alignment map \( \proj \):

\begin{equation}
\label{eqproj}
\pi^{*}_{\varepsilon}(\cdot, \cdot \mid \obsdata)
:= \projsharp \tilde{\pi}_{\varepsilon}(\cdot, \cdot \mid \obsdata)
\end{equation}
where
\[
\proj: (\param, \simdata) \mapsto (\param_{\sigma^*}, \simdata_{\sigma^*}).
\]
By construction, \( \proj \circ \proj = \proj \), so \( \proj \) acts as a projection. Moreover, its preimage satisfies \( \proj^{-1}(\{\param, \simdata\}) = \{(\param_{\sigma}, \simdata_{\sigma}) \mid \sigma \in \mathcal{S}_K\} \).

As a result, the projected pseudo-posterior satisfies:
\begin{equation*}
    \begin{split}
        \pi_{\varepsilon}^* (\param, \simdata \mid \obsdata) &  \propto \Tilde \pi_{\varepsilon} (\proj^{-1}(\{(\param, \simdata)\}) \mid \obsdata)\cdot \mathbb I_{\text{Im}(\proj)}(\param,\simdata)\\
        & \propto \sum_{\sigma} \Tilde \pi_{\varepsilon} (\param_{\sigma}, \simdata_{\sigma} \mid \obsdata)\cdot \mathbb I_{\text{Im}(\proj)}(\param,\simdata)\\
        & \propto \tilde \pi_{\varepsilon}(\param, \simdata\mid \obsdata) \cdot \mathbb I_{\text{Im}(\proj)}(\param,\simdata).
    \end{split}
\end{equation*}

\noindent which enforces a constraint that accepted samples must lie in the image of the projection, i.e., must be optimally aligned with the observed data.

This additional constraint effectively restricts the integration domain to a subset of datasets that are optimally ordered with respect to \( \obsdata \). We define this constrained space as
\[
\dataspace_{\obsdata}^{*K} = \left\{ \simdata \in \dataspace^K \mid d(\obsdata, \simdata) = \min_{\sigma \in \mathcal{S}_K} d(\obsdata, \simdata_{\sigma}) \right\}.
\]
Consequently, the marginal pseudo-posterior can be expressed as:
\begin{equation}
\label{eqpermabc_star}
\pi_{\varepsilon}^*(\param \mid \obsdata) \propto \pi(\param) \int_{A_{\obsdata,\varepsilon}^*} f(\simdata \mid \param) \, d\simdata,
\end{equation}
where the new acceptance region is 
\[
A_{\obsdata,\varepsilon}^* = B_{\varepsilon}(\obsdata) \cap \dataspace_{\obsdata}^{*K}.
\]

The \emph{permABC} algorithm (see Algorithm~\ref{alg:permABC_star}) therefore defines a different approximation of the posterior distribution than classical ABC methods. In particular, it yields a different posterior approximation  \( \pi_{\varepsilon}^*(\cdot \mid \obsdata) \), which we shall compare to the vanilla ABC approximation \( \pi_{\varepsilon}(\cdot \mid \obsdata) \).

\begin{algorithm}[ht]
\caption{permABC Vanilla}
\label{alg:permABC_star}
\textbf{Input: }observed dataset $\obsdata$, number of iterations $N$, threshold $\varepsilon>0$.\\
\textbf{Output: }a sample $(\param^{(1)}, \dots, \param^{(N)})$ from the pseudo-posterior distribution $\pi_{\varepsilon}^*(\cdot\mid \obsdata)$.\\

\BlankLine

\For{$i = 1,\dots, N$}{
    \Repeat{$d(\obsdata, \simdata_{\sigma^*}) \leq \eps$}{
        $\parprop \sim \pi(\cdot)$\;
        $\simdata \sim f(\cdot \mid \parprop)$\;
        $\sigma^* = \argmin_{\sigma \in \mathcal{S}_K} d(\obsdata, \simdata_{\sigma})$\;
    }
    $\param^{(i)} = \param'_{\sigma^*}\;$
}
\end{algorithm}

\subsection{Comparison with ABC}
\label{sec:comparison}

The permABC method introduced in the previous section (Algorithm~\ref{alg:permABC_star}) differs from standard ABC in a key aspect: instead of accepting all permutations that satisfy the ABC criterion, it retains only the optimal one—namely, the permutation that minimizes the discrepancy defined in Equation~\eqref{eqopti}. This choice, motivated by computational efficiency, alters the pseudo-posterior distribution by constraining the pseudo-data to lie within $\dataspace^{*K}_{\obsdata}$, as previously described.

In contrast, a variant equivalent to standard ABC would consist in uniformly sampling a permutation among those satisfying the ABC acceptance condition and appropriately reweighting the contribution of each accepted sample by the size of this set. However, enumerating this set incurs a prohibitive $K!$ computational complexity. In Appendix~\ref{appendix_secstrat}, we propose a stratified sampling strategy—termed \emph{permABC-strat}—to efficiently estimate this cardinality and recover the standard pseudo-posterior.

To clarify the connection between permABC and standard ABC, we establish a sufficient condition under which both methods coincide. Specifically, we prove that when the tolerance $\varepsilon$ is small enough, the set of permutations satisfying the ABC criterion contains at most one element. We formalize this in the following lemma:

\begin{lemma}
\label{lempermABC_star}
Let $\obsdata \in \dataspace^K$ denote the observed dataset, and define the discrepancy function $d(\cdot, \cdot)$ as in Equation~\eqref{eqdistance}. Let
\[
\varepsilon^* := \min_{\sigma \neq \mathrm{Id}} \frac{1}{2} d(\obsdata, \obsdata_\sigma).
\]
Then for any $\varepsilon < \varepsilon^*$ and any simulated dataset $\simdata \in \mathbb{R}^{K \times n}$, at most one permutation satisfies the ABC criterion:
\[
\mathrm{Card} \left( \left\{ \sigma \in \mathcal{S}_K : d(\obsdata, \simdata_\sigma) \leq \varepsilon \right\} \right) \leq 1.
\]
\end{lemma}

The proofs of the preceding lemma, the following theorem, and the subsequent corollary are deferred to Appendix~\ref{appendix_proofs}.

\begin{theorem}
\label{th:permABC_star}
Let $\pi^*_{\varepsilon}( \cdot \mid \obsdata)$ be the pseudo-posterior distribution resulting from the permABC algorithm, and $\pi_{\varepsilon}( \cdot \mid \obsdata)$ that obtained via standard ABC. Then, there exists $\varepsilon^*>0$ such that for all $\varepsilon < \varepsilon^*$, we have:
\[
\pi^*_{\varepsilon}( \cdot \mid \obsdata) = \pi_{\varepsilon}( \cdot \mid \obsdata).
\]
\end{theorem}

This equivalence ensures that permABC retains the same asymptotic behavior as standard ABC in the limit of vanishing tolerance. We summarize this in the following corollary:

\begin{corollary}
\label{col:permabc}
If the pseudo-posterior distribution $\pi_{\varepsilon}(\cdot \mid \obsdata)$ obtained by the ABC algorithm converges to the true posterior distribution as $\varepsilon \to 0$, then the pseudo-posterior distribution $\pi_{\varepsilon}^*(\cdot \mid \obsdata)$ obtained by the permABC algorithm also converges to the true posterior distribution as $\varepsilon \to 0$.
\end{corollary}

Thus, the efficiency gains provided by permABC come at no cost to theoretical consistency. For foundational results on ABC convergence, we refer to \citet{fearnhead2018abc}, who provides general conditions under which ABC approximations converge to the true posterior, along with error bounds.

\section{Sequential strategy with permutations}
\label{sec:sequential}

\subsection{Sequential ABC: state of the art and limitations}
\label{sec:abc_smc}

Rejection-based ABC methods, including the enhanced permABC, remain computationally expensive due to their reliance on the prior as a proposal distribution. 
Most methods in the ABC literature rely instead on more sophisticated proposals. A typical strategy is to propose from some distribution \( q(\param) \), then assign importance weights  of the form \( \pi(\param)/q(\param) \) to accepted samples to correct for the fact that they are not drawn from the prior \( \pi(\cdot) \). This produces weighted samples from the ABC pseudo-posterior \( \pi_\varepsilon(\param \mid \obsdata) \). 

Sequential importance sampling further improves efficiency by constructing a sequence of proposal distributions \( q_t(\cdot) \) over \( T \) iterations. Inspired by the Sequential Monte Carlo (SMC) framework \citep{delmoralSequentialMonteCarlo2006a}, several ABC-SMC algorithms have been developed \citep{sisson2007sequential,sissonCorrectionSissonSequential2009, toniApproximateBayesianComputation2008, beaumontAdaptiveApproximateBayesian2009b, delmoralAdaptiveSequentialMonte2012}, significantly improving convergence and computational performance.

These methods differ in their algorithmic design but follow the same principle: simulate from a sequence of intermediate ABC pseudo-posteriors \( \pi_{\varepsilon_0}, \dots, \pi_{\varepsilon_T} \), where the thresholds satisfy \( \varepsilon_0 = \infty > \varepsilon_1 > \dots > \varepsilon_T = \varepsilon \). Weights are updated iteratively, enabling efficient progression from the prior \( \pi_{\varepsilon_0}(\cdot \mid \obsdata) = \pi(\cdot) \) to the target pseudo-posterior \( \pi_\varepsilon(\cdot \mid \obsdata) \).

Here, we adopt the algorithm proposed by \citet{delmoralAdaptiveSequentialMonte2012}, widely used as a benchmark for ABC methods \citep{clarte2021componentwise}. This method combines a random walk proposal (see Section~\ref{sec:kernels}) with a Metropolis–Rosenbluth-Teller-Hastings (MRTH) acceptance step based on the ABC criterion and the MRTH ratio. 

A notable feature of this algorithm is the possibility to simulate \( M \) datasets \( z^{(1)}, \dots, z^{(M)} \) per parameter value \( \param \). 
While increasing \( M \) can reduce variance in the acceptance decision, setting \( M = 1 \) leads to better computational efficiency, as shown in practice by \citet{delmoralAdaptiveSequentialMonte2012} and in theory by \citet{bornn2017use}. 
With $M=1$, the importance weights become binary: they are equal to 1 if the particle is accepted (i.e., the distance is below \( \varepsilon \)) and 0 otherwise.

This simplification, however, induces a limitation. When weights are binary, the Effective Sample Size (ESS) reduces to the number of surviving particles. While ESS is commonly used to assess sample quality, it may not accurately reflect the true diversity of the sample, especially when the acceptance threshold is low. In such scenarios, repeated resampling steps tend to concentrate the population, leading to particle degeneracy. 
To address this issue, we propose using the \emph{unique particle rate} as a complementary diagnostic. This metric, defined as the proportion of distinct parameter values among the accepted particles, provides a more sensitive indicator of sample diversity. It offers a practical way to monitor degeneracy and evaluate the effective exploration of the pseudo-posterior during the sequential process. 
In contrast, the ABC-PMC algorithm originally introduced by \citet{sissonCorrectionSissonSequential2009} and later corrected by \citet{beaumontAdaptiveApproximateBayesian2009b} uses multinomial resampling, a practice discouraged in the SMC literature due to its inefficiency \citep{chopin2020introduction}.

In the following subsections, we propose three ways to accelerate ABC-SMC with permutations, based on three different strategies to define intermediate distributions.

\subsection{The permABC-SMC algorithm}

The adaptive SMC algorithm proposed by \citet{delmoralAdaptiveSequentialMonte2012} can be seamlessly extended to the permABC framework. By substituting the standard ABC acceptance criterion with the permutation-based criterion \eqref{eqopti}, we develop a sequential algorithm that targets the pseudo-posterior distribution \( \tilde{\pi}_\varepsilon(\param \mid \obsdata) \).

Upon reaching the stopping criterion or when \( \varepsilon_T = \varepsilon \), we apply the projection step \eqref{eqproj} to each particle in the final population \( \{(\param^{T,i}, \simdata^{T,i})\}_{i=1}^N \). This yields samples from the permuted pseudo-posterior \( \pi_\varepsilon^*(\param \mid \obsdata) \), effectively restoring identifiability among compartments.

It is essential to ensure that the proposal distributions \( q_t(\cdot) \) are adapted to the new acceptance criterion. We detail such kernels in Section~\ref{sec:kernels}, emphasizing their role in maintaining the efficiency of the sampling process under the permutation-based framework.

This approach, termed \emph{permABC-SMC}, facilitates a sequential transition from the prior to the permuted pseudo-posterior \( \pi^*_\varepsilon(\param \mid \obsdata) \). Thanks to proposal kernels tailored to the permutation-based selection rule, permABC-SMC enhances the exploration of the parameter space, particularly in high-dimensional settings where standard ABC methods may struggle.

The permABC-SMC approach is a natural way to extend permABC into a sequential method. It relies on intermediate distributions of the form $(\pi^*_{\varepsilon_t})_t$ which are similar to other sequential ABC methods, and it already leads to substantial computational gains, as shown in Section \ref{sec:numeric}. Before we examine these practical results, we introduce two innovative sequences of intermediate distributions, with novel selection criteria which are only available because we rely on permutation-based matching.

\subsection{Over-sampling}
\label{sec:permABC_OS}

The over sampling strategy leverages the flexibility of the permABC framework to sample more efficiently from the target pseudo-posterior \( \pi_\varepsilon^*(\cdot \mid \obsdata) \) by simulating more compartments than are present in the observed data. Specifically, instead of generating \( K \) pseudo-compartments, we simulate \( M > K \) compartments: that is, we consider \( \simdata \in \dataspace^M \) while the observed data remain in \( \dataspace^K \). We then accept or reject based on the best simulated subset of size $K$.

To implement this strategy, we draw \( \param = (\globparam, \locparam_1, \dots, \locparam_M) \sim \pi(\cdot) \), simulate the corresponding data \( \simdata = (z_1, \dots, z_M) \sim f(\cdot \mid \param) \), and use the following acceptance condition:
\begin{equation}
    \label{eq:os}
    \min_{\sigma \in \mathcal{A}_K^M} d(\obsdata, \simdata_{\sigma}) \leq \varepsilon,
\end{equation}
where \( \mathcal{A}_K^M \) denotes the set of injections from \( \{1, \dots, K\} \) to \( \{1, \dots, M\} \), i.e., the set of partial \( K \)-permutations of \( \{1, \dots, M\} \). As in \eqref{eqopti}, this optimization problem can be solved efficiently using the same assignment algorithm with computational complexity \( \mathcal{O}(K^2 M) \) \citep{crouseImplementing2DRectangular2016a}.

We slightly abuse notation by denoting \( f \) the joint simulator over compartments, regardless of whether the number of compartments is \( K \) (observed) or \( M \) (simulated). This is justified since the simulator is defined componentwise via a common mechanism \( g \), applied independently to each compartment: formally, \( f(\simdata \mid \param) := \prod_{m=1}^M g(z_m \mid \globparam, \locparam_m) \).

Analogously to the case \( M = K \), this procedure defines a pseudo-posterior distribution \( \tilde{\pi}_\varepsilon^M(\cdot \mid \obsdata) \), and its projected version \( \pi_\varepsilon^{*M}(\cdot \mid \obsdata) \), which is obtained by applying the optimal partial permutation minimizing \eqref{eq:os} in a reasoning analogous to that of Section~\ref{secabc_to_permabc} (see Appendix~\ref{appendix_os} for details).

Interestingly, the marginals in $\globparam$ and $\locparam$ exhibit contrasting behaviors as $M$ varies. When $M$ is large, the acceptance condition becomes easier to satisfy regardless of the value of $\varepsilon$, leading to a joint pseudo-posterior distribution close to the prior. Since the permutation only affects the local parameters, the marginal over $\globparam$ remains nearly unchanged, and therefore close to its prior distribution. In contrast, the marginal over the local parameters becomes more concentrated. This is because, when $M$ is large, the optimization in \eqref{eq:os} explores a wider set of potential matches, and the local parameters $\locparam_{1:K}$ corresponding to the best matching compartments are those that generate data most similar to $\obsdata$ under a given $\globparam$. As a result, these local components become tightly clustered around values that are locally optimal with respect to the matching criterion. As $M$ decreases toward $K$, this concentration fades, and the marginal over $\locparam$ gradually approaches the full permABC pseudo-posterior. This behavior is illustrated in Figure~\ref{fig:os}, and further discussed in Appendix~\ref{appendix_os}.

This new distribution family is particularly appealing, as it introduces a natural sequence of auxiliary distributions that converge to the target distribution \( \pi_\varepsilon^*(\cdot \mid \obsdata) \). We define a decreasing sequence \( M_0 > M_1 > \cdots > M_T = K \), keeping \( \varepsilon \) fixed, such that the acceptance condition \eqref{eq:os} becomes increasingly selective until it recovers the standard permABC criterion \eqref{eqopti} when \( M = K \). 
To navigate this sequence \( \left \{\tilde{\pi}_\varepsilon^{M_t}, t = 0, \dots, T \right\} \), we adapt the ABC-SMC framework of \citet{delmoralAdaptiveSequentialMonte2012}, as discussed in Section~\ref{sec:abc_smc}. At each stage, we apply the optimal partial permutation \( \sigma^* \) to each accepted particle to obtain a sample from \( \pi_\varepsilon^*(\cdot \mid \obsdata) \). This results in a new algorithm we call \emph{permABC-Over Sampling (permABC-OS)}.

This method  allows convergence to the target distribution in a reduced number of iterations. However, it also introduces several challenges that we now discuss. The first requirement for the success of the method is to initialize the algorithm with a sufficiently large and diverse population of accepted particles. To achieve this, one must choose $M_0$ and $\varepsilon$ in a coordinated manner. Specifically, since the minimum in \eqref{eq:os} decreases with $M$, a smaller value of $\varepsilon$ necessitates a larger $M_0$ in order to maintain a reasonable acceptance rate at the first iteration. In Appendix~\ref{appendix_os}, we propose a method to calibrate $\varepsilon$ as a function of $M_0$, thereby allowing users to set $M_0$ under a fixed simulation budget.

A second challenge attached to this method is the risk of extinction of the particle population during the sequential iterations. Specifically, when transitioning from population $t$ with $M_t$ compartments to population $t+1$ with $M_{t+1} < M_t$, the acceptance condition \eqref{eq:os} is evaluated on a reduced subset of simulated compartments. If compartments that previously matched observed ones are no longer considered, the particle may be rejected—even though it was previously accepted. To mitigate this, the sequence $M_t$ must be selected carefully. We provide a practical heuristic for this in Appendix~\ref{appendix_os}. 

Even under a careful design of the sequence $M_t$, a large portion of the population may be lost in this transition. To address this issue, we introduce a duplication mechanism whereby each particle is replicated across multiple random permutations. This increases the probability that at least one version of the particle survives the next acceptance test, without requiring additional simulations. The duplication strategy, along with heuristic guidelines for both selecting the sequence \( (M_t) \) and tuning \( \varepsilon \), is also detailed in Appendix~\ref{appendix_os}. 

\begin{figure}[ht]
    \centering
    \includegraphics[scale=.5]{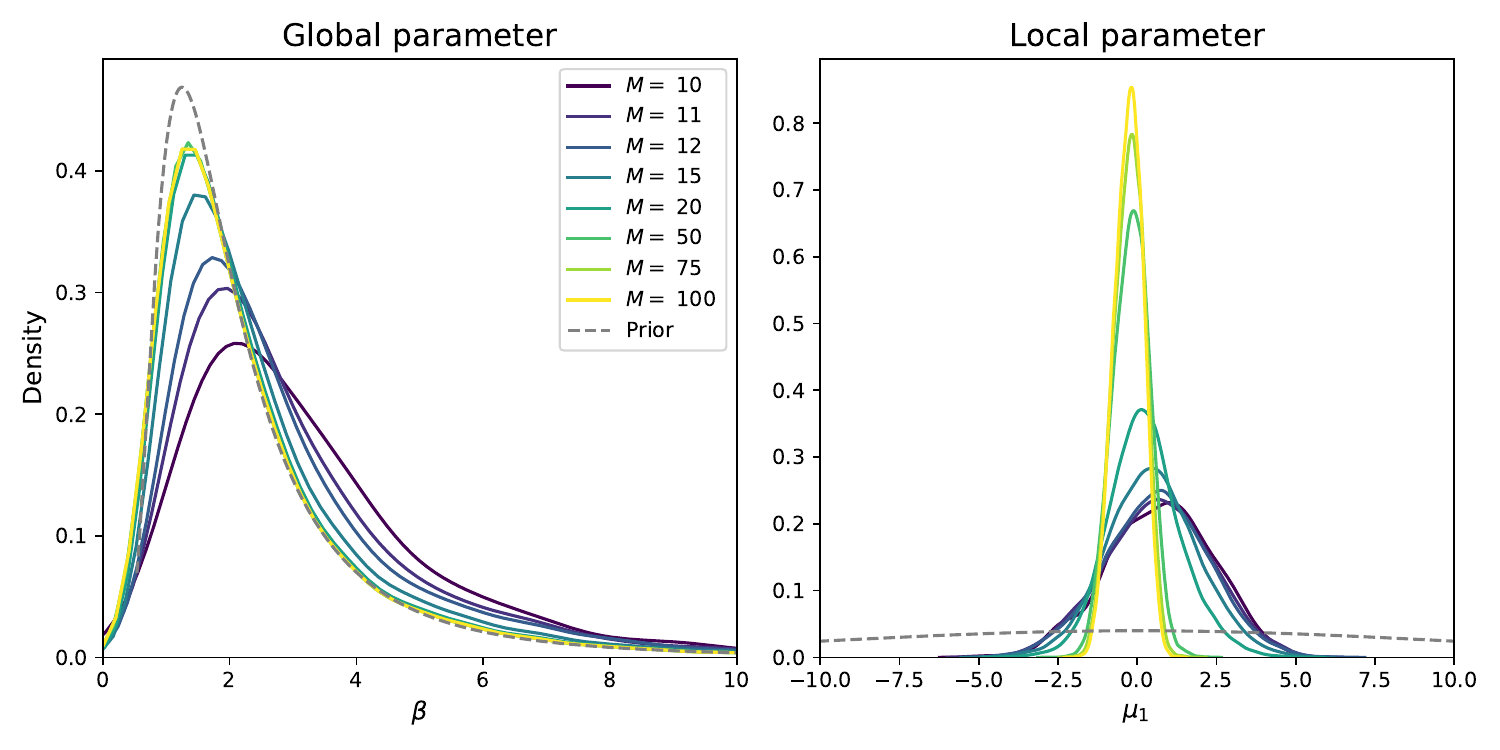}
    \caption{Evolution of the pseudo-posterior \( \pi_\varepsilon^{*M} \) as \( M \) decreases from \( M_0 = 150 \) (yellow) to \( M_T = K = 10 \) (purple) for a Gaussian toy model.}
    \label{fig:os}
\end{figure}

\subsection{Under-matching}
\label{sec:undermatching}

The under matching strategy—referred to as \emph{permABC-Under Matching (permABC-UM)}—adopts a  perspective opposite to over sampling. 
In this approach, we simulate synthetic data with the same number $K$ of compartments as the observed data; but at intermediate steps, we no longer require the simulated data to match all $K$ observed compartments. Instead, we simulate $K$ compartments but only require $L<K$ to match. The resulting acceptance criterion is defined as:
\begin{equation}
    \label{eqopti_um}
    \min_{\sigma \in \mathcal{A}_L^{K}, \tilde{\sigma} \in \mathcal{A}_L^{K}} d(\obsdata_{\tilde \sigma}, \simdata_{\sigma}) \leq \varepsilon,
\end{equation}
where \( \sigma \) and \( \tilde{\sigma} \) denote the sets of \( L \) selected compartments from the simulated and observed datasets, respectively. 

To solve this optimization problem efficiently, we extend the cost matrix by adding \( K - L \) dummy rows and \( K - L \) dummy columns with zero cost. These artificial entries ensure that \( K - L \) assignments are trivially matched, leaving the algorithm to find the best \( L \) true matches. This reduces the partial matching problem to a standard linear sum assignment problem on a square matrix of size \( (2K - L) \times (2K - L) \), which can be solved in \( \mathcal{O}\{(2K-L)^3\} \) time using classical algorithms \citep{crouseImplementing2DRectangular2016a}. 

This formulation allows for the algorithm to search over all possible \( L \)-sized subsets in both datasets, thereby relaxing the matching constraint and increasing flexibility in early iterations.

Although the distance \( d \) in Equation~\eqref{eqdistance} is defined for full datasets in \( \dataspace^K \), it naturally extends to subsets. In this context, given \( \sigma, \tilde \sigma \in \mathcal{A}_L^K \), the restricted distance is such that
\[
d(\obsdata_{\tilde \sigma}, \simdata_{\sigma})^2 = {\sum_{k=1}^L w_{\tilde \sigma(k)}^2 \| \obsdata_{\tilde \sigma(k)} - \simdata_{\sigma(k)} \|_2^2}.
\]
This formulation implies that compartments with higher weights have a greater impact on the acceptance condition, and are thus harder to match. This property can be exploited to prioritize compartments deemed more informative or robust, by appropriately adjusting the weights \( w_k \).

As in other variants of permABC, we employ a SMC scheme through a sequence of intermediate distributions \( \tilde{\pi}_\varepsilon^{L_t} \), corresponding to an increasing schedule \( L_0 < L_1 < \dots < L_T = K \). This gradual tightening of the matching criterion provides a smooth path toward the final target \( \pi^*_\varepsilon(\cdot \mid \obsdata) \), where all compartments are required to match.

Nevertheless, permABC-Under Matching faces challenges similar to those of the Over Sampling approach. In particular, it is crucial to select the initial tolerance \( \varepsilon \) appropriately, so as to ensure a sufficiently large and diverse first population. This requires calibrating \( \varepsilon \) as a function of \( L_0 \), the initial number of matched compartments, towards guaranteeing a high survival rate at the start of the algorithm.

Moreover, when the model is misspecified for even a single observed compartment, the final transition from \( L_{T-1} = K - 1 \) to \( L_T = K \) can lead to a sharp drop in the acceptance rate, resulting in severe particle degeneracy or even extinction. This makes the design of the sequence \( (L_t) \) and the level \( \varepsilon \) particularly critical for the success of the method. 
We show in Section \ref{sec:osum} that a well-designed under matching approach can increase performance in the presence of outliers.

These aspects are discussed in Appendix~\ref{appendix_um}, along with practical guidelines for choosing \( L_t \) and calibrating \( \varepsilon \) to avoid degeneracy. Situations where permABC-Under Matching is especially advantageous, such as in the presence of outliers or structural model misspecification, are further investigated in Section~\ref{sec:osum}.

\subsection{Proposal Kernels for Permuted Inference}
\label{sec:kernels}

In the ABC-SMC framework of \citet{delmoralAdaptiveSequentialMonte2012}, the MCMC kernel \( K_t \) at iteration \( t \) uses a Gaussian random walk proposal \( q_t \) of the form:
\[
q_t(\cdot \mid \param) = \mathcal{N}\bigl\{\param, \operatorname{diag}\bigl(\tau_t^2\bigr)\bigr\},
\]
where \( \tau_t^2 \in \mathbb{R}^d \) is a vector of componentwise variances. These are typically defined as:
\[
\tau_t^2 = 2 \cdot \operatorname{var}\bigl\{\bigl(\param^{t-1,i}\bigr)_i\bigr\},
\]
i.e., twice the empirical variance across particles from the previous population. A Metropolis–Hastings correction step is then applied using the ABC acceptance criterion.

In the case of permABC-SMC, the exchangeability of the local parameters introduces label switching in the particle populations, which complicates the direct computation of componentwise variances. To resolve this, we first align the local parameters across particles using their optimal permutations \( \sigma_i^* \), and compute:
\[
\tau_t^2 = 2 \cdot \operatorname{var}\bigl\{\bigl(\param_{\sigma_i^*}^{t-1,i}\bigr)_i\bigr\}.
\]
This alignment restores parameter consistency across compartments, and allows us to maintain a diagonal Gaussian kernel with meaningful componentwise variances. The proposal kernel is then given by:
\[
q_t(\cdot \mid \param, \sigma) = \mathcal{N}\bigl[\param, \operatorname{diag}\bigl\{\sigma^{-1}(\tau_t^2)\bigr\}\bigr],
\]
where \( \sigma^{-1}(\tau_t^2) \) permutes the variance vector accordingly, so that each coordinate of \( \param \) is perturbed using the variance estimated in the appropriate reference frame.

When the number of simulated compartments \( M \) exceeds the number of observed compartments \( K \), only a subset of the local parameters are matched to observations at each iteration. Let \( \mathcal{I} = \mathrm{Im}(\sigma^*) \subset \{1, \dots, M\} \) denote the indices of simulated compartments that are matched. For those indices, we reuse the corresponding variances from \( \tau_t^2 \). For the unmatched components \( k \notin \mathcal{I} \), we assign a default value:
\[
\tau_0^2 := \max_{k \in \mathcal{I}} \tau_t^2[k],
\]
which ensures exploration and numerical stability. This same default value \( \tau_0^2 \) is reused in the under-matching case below.

In the case of under matching, only a subset \( L_t < K \) of observed compartments are matched at each iteration. Let \( \sigma \in \mathcal{A}_L^K \) and \( \tilde \sigma \in \mathcal{A}_L^K \) be the injections representing matches from observed to simulated compartments and vice versa. For each observed coordinate \( k \in \{1,\dots,K\} \), we define:
\[
\tau_t^2[k] = 2 \cdot \operatorname{var}\bigl\{\bigl(\theta_{\tilde \sigma_i^{-1} \circ \sigma_i(k)}^{t-1,i}\bigr)_i\bigr\},
\]
whenever the compartment \( k \) has been matched frequently enough across the population. If this condition is not met (e.g., in early iterations or due to sparse matching), we fallback to the default value \( \tau_t^2[k] = \tau_0^2 \).

These variance assignment strategies allow us to construct robust and coherent Gaussian random walk proposals \( q_t \) for both global and local parameters across all three regimes: standard permABC, Over Sampling, and Under Matching.

\section{Numerical experiments}
\label{sec:numeric}

\subsection{Comparison between ABC and permABC}
\label{sec:abc_vs_permabc}

As discussed in Section~\ref{sec:comparison}, permABC exhibits a different behavior from standard ABC. In particular, two regimes emerge depending on the value of the threshold \( \varepsilon \): when \( \varepsilon > \varepsilon^* \), permABC yields a different approximation of the posterior compared to ABC; whereas when \( \varepsilon \leq \varepsilon^* \), permABC coincides with standard ABC and thus inherits its convergence guarantees.

To better illustrate the differences between the two approaches, we consider a simple toy model with no global parameter, defined as follows:
\[
\mu_k \sim \mathcal{U}(-2, 2), \quad \obsdata_k \sim \mathcal{U}(\mu_k - 1, \mu_k + 1), \quad k = 1, \dots, K,
\]
with \( K = 2 \), and a fixed observation \( y = (-1, 1)^\top \). 

Figure~\ref{fig:unif} highlights the two different regimes. When \( \varepsilon > \varepsilon^* \), permABC is concentrated in a region constrained by the order statistics, leading to a pseudo-posterior different from that of standard ABC. However, when the threshold satisfies \( \varepsilon \leq \varepsilon^* \), both methods recover the same pseudo-posterior. In this example, the critical value is \( \varepsilon^* = \sqrt{1/8}/2 \approx 0.177 \).

Note that, even in its vanilla form, permABC offers a significant gain in simulation efficiency. Since it effectively avoids redundant permutations during rejection, it reduces the required number of simulations by a factor of approximately \( K! \). In this simple example with \( K = 2 \), we achieve the same approximation quality (i.e., same \( \varepsilon < \varepsilon^* \)) using only half the number of simulations compared to standard ABC.

\begin{figure}[ht]
    \centering
    \includegraphics[scale=.6]{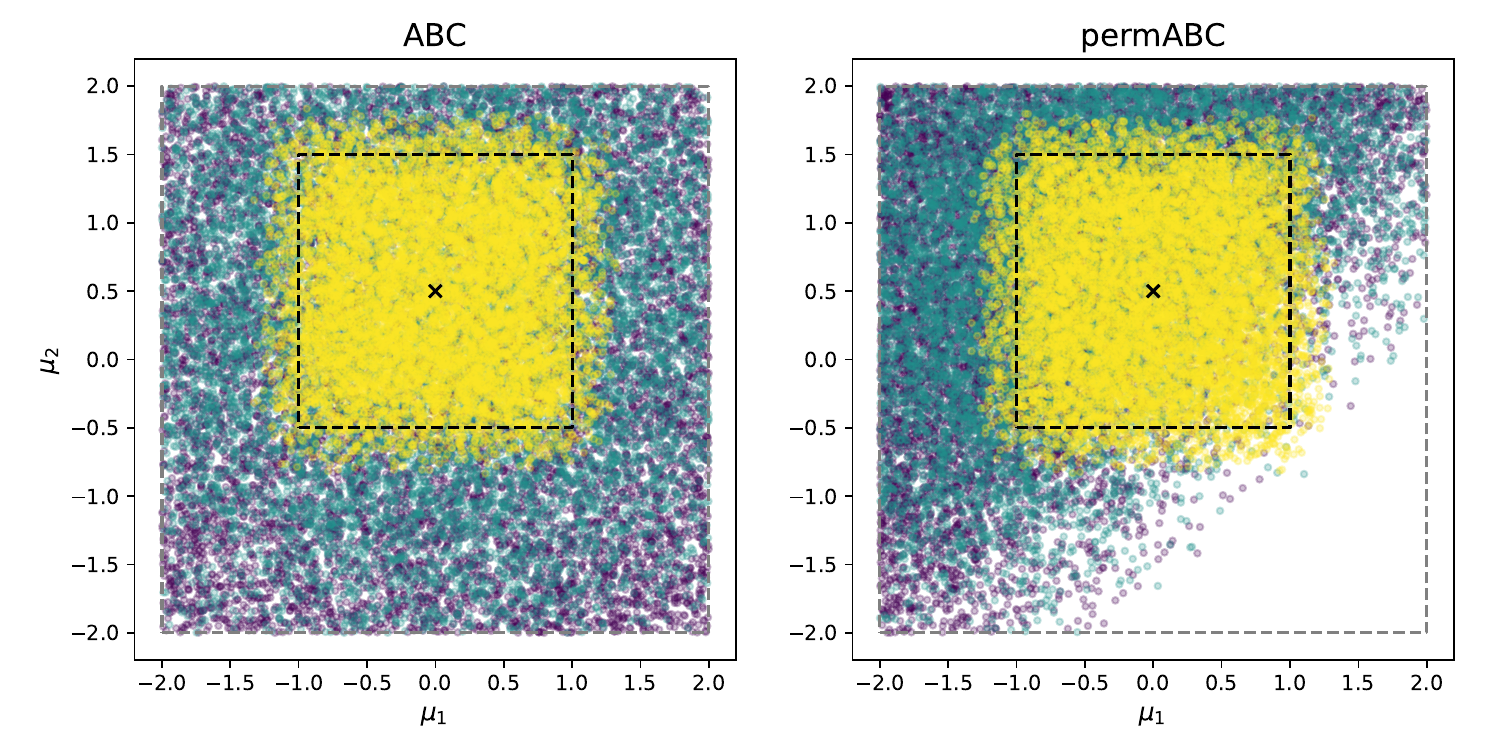}
    \caption{
        Comparison between standard ABC (left) and permABC (right) on the toy model of Section \ref{sec:abc_vs_permabc} with \( K = 2 \) and observation \( y = (-1, 1)^\top \). Each dot represents a sampled configuration \( \mu = (\mu_1, \mu_2) \), colored by the acceptance threshold used: purple for \( \varepsilon = \infty \), green for \( \varepsilon > \varepsilon^* \), and yellow for \( \varepsilon = \varepsilon^* \). The black cross denotes the observation \( y \), and the dashed box indicates the true posterior. While both methods recover the same posterior when \( \varepsilon < \varepsilon^* \), permABC avoids simulating permutations that are suboptimal at large \( \varepsilon \), resulting in a more concentrated distribution and improved efficiency. The two panels use the same computational budget.
    }
    \label{fig:unif}
\end{figure}

\subsection{Challenges of High-Dimensional ABC: A Gaussian Toy Model}

We now consider a classical hierarchical model commonly used in the Bayesian literature:
\begin{equation*}
    \label{eqgauss}
    \beta \sim \mathrm{IG}(a,b), \quad \mu_k \sim \mathcal{N}(0, s^2), \quad \obsdata_{k} = (y_{1,k}, \dots, y_{n,k}) \overset{\text{i.i.d.}}{\sim} \mathcal{N}(\mu_k, \beta), \quad k = 1, \dots, K.
\end{equation*}

This simple yet representative model provides a convenient testbed to compare the various methods discussed in this paper, and to highlight the limitations of standard ABC approaches in high-dimensional settings. All ABC methods involve a fundamental trade-off between the number of simulations \( N \), the tolerance threshold \( \varepsilon \), and the effective number of accepted samples (or the number of unique particles in sequential schemes). A more efficient method achieves a lower value of \( \varepsilon \) for a given simulation budget \( N \) and fixed final sample size.

To enable fair comparison across methods, we adopt a consistent evaluation criterion: for each experiment, we report the value of \( \varepsilon \) achieved when the algorithm yields exactly 1,000 unique particles. In this setting, the number of simulations required to obtain a single unique accepted sample serves as a proxy for the effective sample size (ESS). This pseudo-ESS provides a normalized measure of sampling efficiency, which we use throughout our benchmarks.

Rejection-based algorithms such as vanilla ABC, as well as sequential variants like ABC-SMC and ABC-PMC, are known to scale poorly with the dimension of the parameter space, here indexed by the number of groups \( K \). As \( K \) increases, the acceptance rate drops sharply, resulting in a dramatic rise in computational cost. 
In contrast, the permutation-based acceptance rule introduced in permABC significantly improves the efficiency of both rejection and sequential samplers in high-dimensional regimes. 
The symmetry structure across local parameters \( \mu_1, \dots, \mu_K \) allows permABC to reduce redundant simulations and produce sharper approximations of the target pseudo-posterior.

Figure~\ref{fig:vanilla_vs_smc} illustrates this phenomenon: the methods proposed in this work achieve lower Monte Carlo error levels using fewer simulations, as reflected by smaller effective thresholds \( \varepsilon \). These gains become increasingly pronounced as \( K \) grows, confirming the advantage of permutation-based strategies in hierarchical or exchangeable models.

\begin{figure}[ht]
    \centering
    \includegraphics[scale=.5]{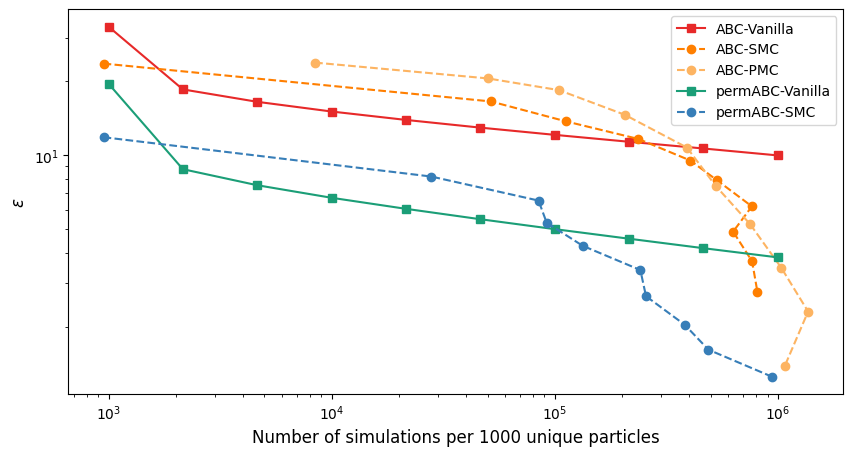}
    \caption{
        Comparison of standard ABC methods and permutation-based variants in terms of achieved tolerance \( \varepsilon \) versus computational budget. The \( x \)-axis shows the number of simulations per 1,000 unique accepted particles (log-scale), and the \( y \)-axis represents the effective tolerance \( \varepsilon \) (log-scale). Lower values of \( \varepsilon \) correspond to more accurate approximations of the pseudo-posterior. Each curve represents a different algorithm. Permutation-based methods achieve substantially lower tolerance levels for the same simulation budget.
    }
    \label{fig:vanilla_vs_smc}
\end{figure}

\subsection{Failure cases of ABC-Gibbs}
\label{sec:abcgibbs_failure}

ABC-Gibbs methods \citep{clarte2021componentwise,rodrigues2020likelihood} have been designed to handle hierarchical models similar to those we consider. They are however less robust to model dependencies. We illustrate this with a modified hierarchical model designed to highlight the limitations of ABC-Gibbs:
\begin{equation*}
    \label{eqgauss_robin}
    \beta \sim \mathcal{N}(0, 10^2), \quad \mu_k \sim \mathcal{N}(0, 10^2), \quad \obsdata_k = (\obsdata_{1,k}, \dots, \obsdata_{n,k}) \overset{\text{i.i.d.}}{\sim} \mathcal{N}(\mu_k + \beta, 1), \quad k = 1, \dots, K.
\end{equation*}
This model is deliberately over-parameterized: it includes \( K+1 \) parameters for only \( K \) observations. Despite its simplicity, the posterior distribution exhibits a strong ridge structure due to the symmetry between local parameters \( \mu_k \) and the global parameter \( \beta \). This posterior can be reliably approximated using reference MCMC software such as PyMC. 

Figure~\ref{fig:perm_vs_gibbs} shows the two-dimensional marginal posterior over \( (\mu_1, \beta) \) obtained under a fixed simulation budget. While permABC-SMC provides a conservative and coherent approximation of the target posterior—as expected for an ABC method with global proposals—ABC-Gibbs produces a truncated approximation, failing to explore the full posterior support.

This discrepancy is rooted in the very structure of ABC-Gibbs. The algorithm proceeds by sequentially updating each parameter conditionally on the others using ABC rejection steps. In the presence of strong dependencies between parameters, as is the case here, the algorithm exhibits poor mixing: each marginal update is constrained by the current (possibly incorrect) values of the other parameters. As a result, the sampler explores a narrow region of the parameter space, failing to recover the full shape of the posterior.

These results underscore the limitations of coordinate-wise ABC approaches in models with strong local-global coupling. In contrast, permABC-based strategies, which jointly consider the full dataset and parameter configuration at each step, are more robust to such structural challenges.

\begin{figure}[ht]
    \centering
    \includegraphics[scale=.6]{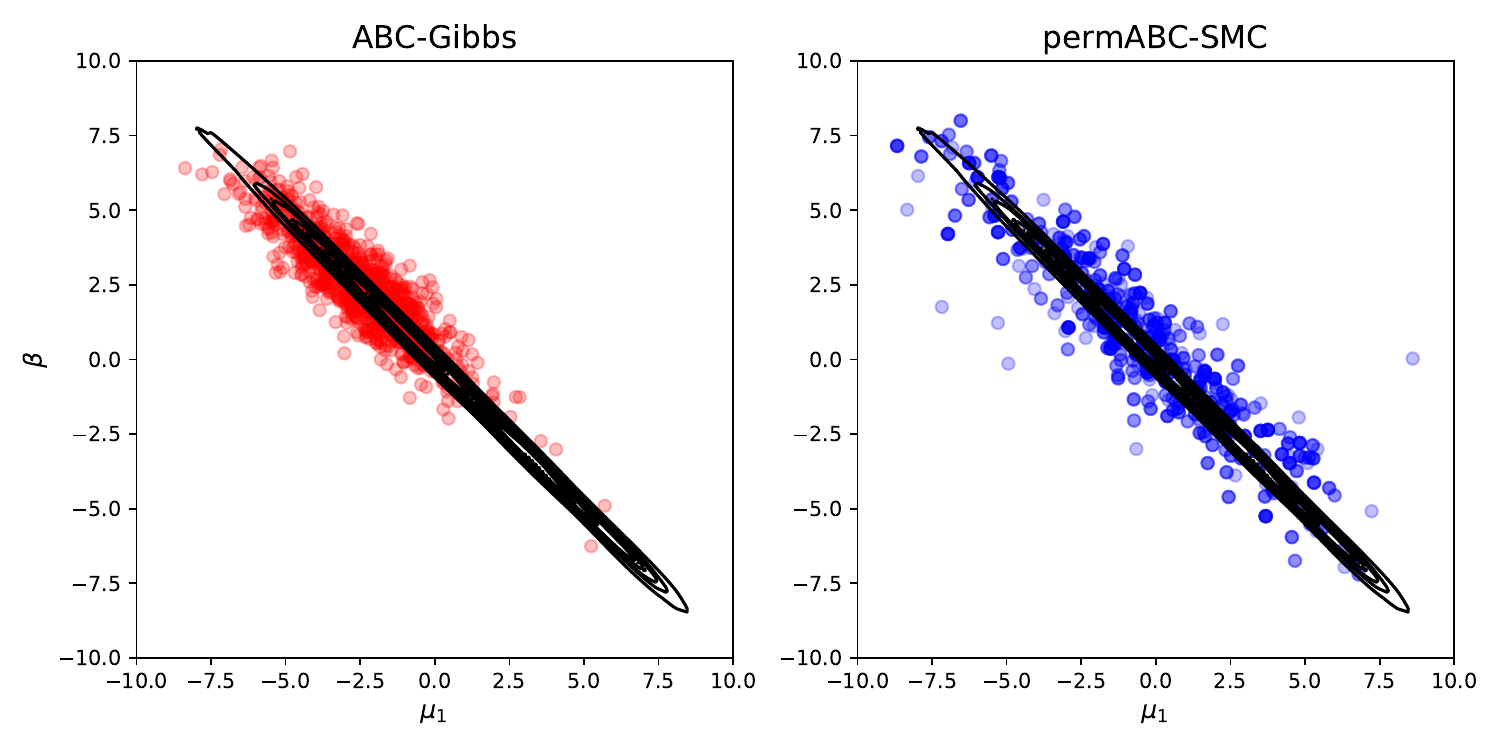}
    \caption{
        Two-dimensional marginal posterior over \( (\mu_1, \beta) \) for the hierarchical model defined in Equation~\eqref{eqgauss_robin}, under a fixed simulation budget. The black ellipses represent level sets of the exact posterior distribution computed via MCMC (PyMC). Blue points correspond to samples produced by permABC-SMC, and red points to those produced by ABC-Gibbs. While permABC-SMC captures the full posterior geometry and respects the strong negative correlation between \( \mu_1 \) and \( \beta \), ABC-Gibbs fails to explore the posterior support. It becomes trapped along a narrow subset of the ridge, leading to a truncated and biased approximation. This illustrates the poor mixing behavior of coordinate-wise ABC samplers in the presence of strong global-local dependencies.
    }
    \label{fig:perm_vs_gibbs}
\end{figure}

\subsection{Discussion of the Over Sampling and Under Matching Approaches}
\label{sec:osum}

The Over Sampling and Under Matching strategies provide two complementary extensions of the permABC framework, each improving robustness and efficiency in different regimes of hierarchical inference under exchangeability.

We evaluate their performance on a Gaussian toy model with \( K = 20 \) compartments, among which 4 have local means \( \mu_k \) lying in the tails of the prior. This setup simulates a realistic scenario involving mild model misspecification or contamination by outliers.
Figure~\ref{fig:osum} reports the number of simulations required to reach a given ABC threshold \( \varepsilon \), normalized by 1000 unique particles.

The over sampling method (permABC-OS) achieves lower values of \( \varepsilon \) with fewer iterations than standard permABC-SMC. While each iteration involves more expensive computations—due to larger assignment problems and duplicated particles (see Appendix~\ref{appendix_os})—the faster convergence compensates for this cost, making the method highly effective even in the low-tolerance regime.

Under matching (permABC-UM) performs well in the intermediate-\( \varepsilon \) range. By relaxing the matching constraint to only \( L < K \) compartments, it allows rapid convergence while retaining robustness against model misspecification. However, as the algorithm transitions to full matching (\( L_t \to K \)), the acceptance condition tightens, often resulting in degeneracy of the particle population. Unlike in the OS strategy, this issue cannot be alleviated by duplication, and the method typically fails to reach very small values of \( \varepsilon \).

These contrasting behaviors suggest that combining the two strategies could yield additional benefits. One possibility is to begin inference with permABC-UM to quickly eliminate poorly fitting regions, then switch to permABC-SMC to refine toward smaller tolerances. More generally, hybrid schemes where \( L_t \) and \( M_t \) are jointly adapted, potentially alternating with changes in \( \varepsilon \), offer a promising direction for future exploration.

\begin{figure}[ht]
    \centering
    \includegraphics[scale=.5]{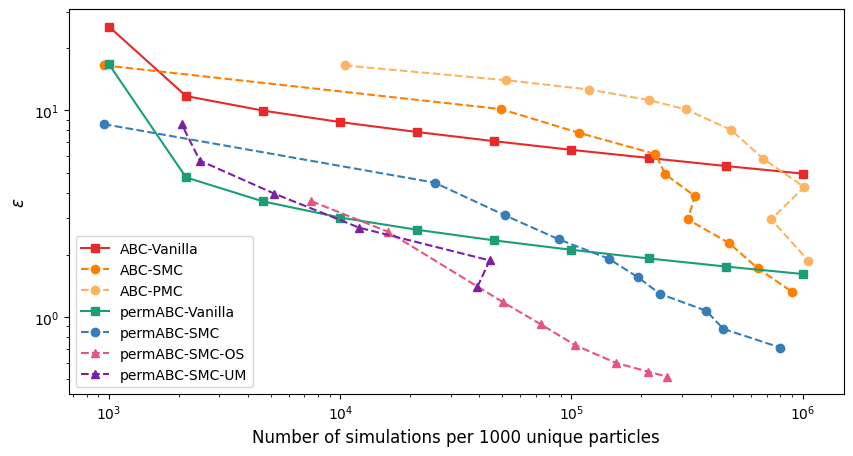}
    \caption{
        Comparison of standard ABC methods and permutation-based variants in terms of achieved tolerance \( \varepsilon \) versus computational budget. The \( x \)-axis shows the number of simulations per 1,000 unique accepted particles (log-scale), and the \( y \)-axis represents the effective tolerance \( \varepsilon \) (log-scale). Lower values of \( \varepsilon \) correspond to more accurate approximations of the pseudo-posterior. Each curve represents a different algorithm. Permutation-based methods achieve substantially lower tolerance levels for the same simulation budget.}
        
    \label{fig:osum}
\end{figure}

\subsection{Application to Real Data: The SIR Model}
\label{sec:sir}

In this final section, we demonstrate the applicability of our methodology on real-world data by analysing the early spread of COVID-19 in France, using publicly available hospitalization records from Santé publique France \citep{DataCovid}. The dataset provides daily counts of hospital admissions at various geographic levels. We focus on mainland France (excluding overseas territories and Corsica) resulting in 94 departmental trajectories of infected individuals.

To model the epidemic dynamics, we use the classical Susceptible–Infectious–Recovered (SIR) model, governed by the following system of ordinary differential equations:
\begin{equation}
    \label{eqSIR}
    \begin{split}
        \frac{dS}{dt} &= -\delta \frac{SI}{N},\\
        \frac{dI}{dt} &= \delta \frac{SI}{N} - \nu I,\\
        \frac{dR}{dt} &= \nu I,
    \end{split}
\end{equation}
where \( S \), \( I \), and \( R \) represent the numbers of susceptible, infectious, and recovered individuals, respectively. The parameters \( \delta \) and \( \nu \) denote the transmission and recovery rates, and the basic reproduction number is defined as \( R_0 = \delta / \nu \). 
To mitigate confounding factors such as viral mutations or heterogeneous regional interventions, we restrict our analysis to the first wave of the epidemic, from March to July 2020. During this early stage, the virus strain and public health policies were relatively homogeneous across the country, justifying a simplified hierarchical modeling approach.

Initial explorations of the data (see Appendix~\ref{appendix_sir_validation}) suggest that the reproduction number \( R_0 \) is shared across all departments: we treat it as a global parameter. Each department \( k =  1, \dots, 94 \) is associated with local parameters \( \locparam_k = (I_k(0), R_k(0), \delta_k) \), describing the initial conditions and the local transmission rate. This formulation naturally fits within the permABC framework, with exchangeability and conditional independence across geographic compartments. When defining the distance between observed and simulated data (see \eqref{eqdistance}), the weights \( w_k \) may optionally be chosen proportional to the population of each department, to account for the relative importance of each trajectory.

This epidemiological context was our original motivation for developing the permABC methodology. It exemplifies a setting where classical ABC approaches often fail due to identifiability issues and the lack of permutation invariance. Our framework was specifically designed to address these challenges in a scalable and robust manner.

Full Bayesian inference at the departmental level leads to a high-dimensional parameter space: with \( K = 94 \) compartments, three local parameters per compartment, and a single global parameter \( R_0 \), the total dimension reaches \( d = 3 \times 94 + 1 = 283 \). Such settings are intractable for standard ABC methods, even in their sequential variants, due to the curse of dimensionality and inefficient exploration. 
In contrast, permABC-SMC enables efficient inference in this setting, especially when combined with the independent-proposal Metropolis-within-Gibbs kernels introduced in Section~\ref{sec:kernels}. It captures both global trends and local heterogeneity, allowing joint inference across all departments while preserving precision on the global reproduction number.

Figure~\ref{fig_SIR_posterior} illustrates this advantage by comparing the posterior distributions of \( R_0 \) obtained using either aggregated national data or the full set of departmental trajectories. When using national data, the inference yields flatter posteriors due to the loss of spatial information. In contrast, the posterior obtained via permABC-SMC using all departmental data is significantly more concentrated, highlighting the benefits of hierarchical and permutation-aware inference. This improvement stems from the model's ability to match each trajectory with different local parameters \( \mu_k \). These local heterogeneities are visualized on the right panel of Figure~\ref{fig_SIR_posterior}, and also mapped geographically in Figure~\ref{fig_SIR_map}.

However, the SIR model remains a simplified representation of epidemic dynamics, and is naturally misspecified when applied to complex real-world trajectories. This results in posterior estimates of the recovery rate \( \nu_k \) that can reach values as high as 4, which are not epidemiologically plausible given that \( \nu \) is a rate that should not typically exceed 1. This limitation illustrates the trade-off between model simplicity and fidelity in real-data applications.

\begin{figure}[ht]
    \centering
    \includegraphics[scale=.6]{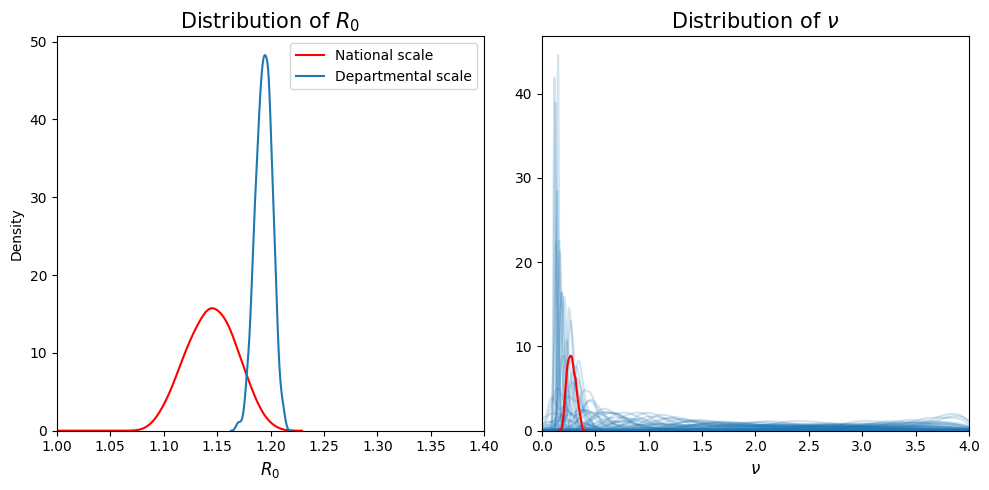}
    \caption{Posterior approximation for the global \(R_0\) parameter and the local recovery rates \(\nu_k\) in the hierarchical SIR model described in \(\eqref{eqSIR}\), where \(R_0 = \delta / \nu\). The inference is performed using departmental-level data (\( K = 94 \) in blue) or national-level aggregation (\(K = 1\) in red).}
    \label{fig_SIR_posterior}
\end{figure}

\begin{figure}[ht]
    \centering
    \includegraphics[scale=.6]{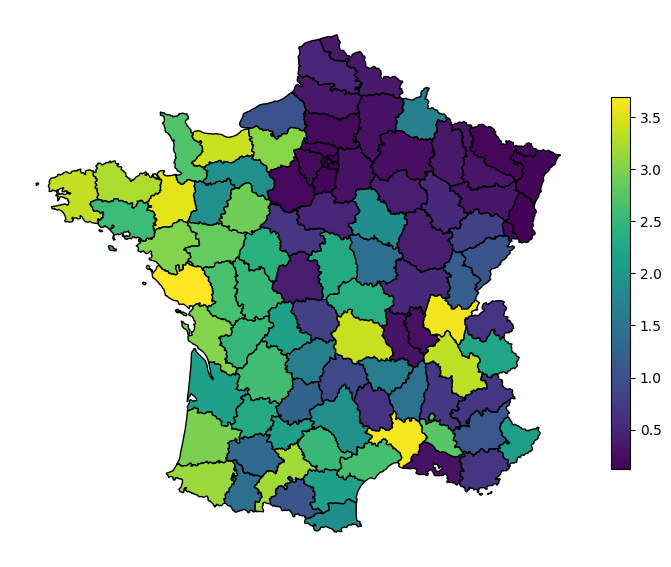}
    \caption{Map of the posterior mean estimates of the local recovery rate \(\nu_k\) for each department.}
    \label{fig_SIR_map}
\end{figure}

\section{Conclusion}

We have introduced \emph{permABC}, a novel framework for Approximate Bayesian Computation (ABC) tailored to hierarchical models with exchangeable structure. By exploiting the invariance of compartments under permutations, permABC restores identifiability among local parameters and significantly improves acceptance rates in high-dimensional settings.

Extending this framework, we developed a sequential version, \emph{permABC-SMC}, which adapts classical ABC-SMC algorithms to permutation-based acceptance rules. We further proposed two complementary strategies: \emph{Over Sampling}, which enriches the simulation space early on to facilitate acceptance, and \emph{Under Matching}, which relaxes the matching constraint to allow robust early-stage inference. These approaches provide new tools for navigating pseudo-posterior sequences, enabling both computational efficiency and resilience to model misspecification.

Our theoretical results show that permABC retains the desirable convergence properties of classical ABC algorithms in the small-tolerance regime. Furthermore, numerical experiments highlight the efficiency and stability of permABC in high-dimensional scenarios, as well as its robustness to model misspecification where coordinate-wise approaches like ABC-Gibbs may fail.

Finally, we demonstrated the practical relevance of our method through an application to COVID-19 epidemic data, using a hierarchical SIR model across 94 departments. This application highlights the practical impact of the permABC: leveraging local heterogeneity leads to sharper inference on global parameters, opening new directions for scalable and robust likelihood-free inference.

All proposed methods are implemented in a dedicated Python package called \texttt{permabc}, available at \url{https://github.com/AntoineLuciano/permABC}. The repository also includes code to reproduce the figures and experiments presented in this paper.

Future directions include a deeper theoretical and empirical investigation of the Over Sampling and Under Matching strategies, with particular emphasis on the design of more efficient proposal distributions. An additional promising avenue is to explore joint or adaptive use of both strategies within a unified sequential framework: their complementarity may further enhance inference performance.

\section*{Acknowledgements}

Antoine Luciano is supported by a PR[AI]RIE PhD grant. Christian P. Robert is also affiliated with the University of Warwick, UK, and ENSAE-CREST, Palaiseau, and is funded by the European Union under the GA 101071601, through the 2023-2029 ERC Synergy grant OCEAN and by a PR[AI]RIE chair from the Agence Nationale de la Recherche (ANR-19-P3IA-0001).

\bibliography{permABC}

\appendix

\section{Stratified version of permABC}
\label{appendix_secstrat}

Even though permABC is equivalent to classical ABC in the small $\varepsilon$ regime, it differs outside of this regime because it integrates $\simdata$ into the constrained space $\dataspace_{\obsdata}^*$. Here, we propose a method called stratified permABC, which maintains the same acceptance rate but is exactly equivalent to ABC regardless of the value of $\varepsilon$.

In the stratified permABC algorithm, similar to the standard approach, we sample from the prior $\parprop \sim \pi(\cdot)$ and then sample data from the likelihood: $\simdata \sim f(\cdot \mid \parprop)$. Here, we consider the set of all permutations in the symmetric group $\mathcal{S}_K$, and determine the subset of acceptable permutations $\mathcal S_K^\varepsilon = \{ \sigma \in \mathcal S_K : d(\obsdata,\simdata_{\sigma}) \leq \varepsilon \}$. The acceptance condition is then $\left | \mathcal S_K^\varepsilon \right | > 0$, which is equivalent to $\min_{\sigma \in \mathcal S_K} d(\obsdata,\simdata_{\sigma}) \leq \varepsilon$. We then accept the simulated $\parprop_{\sigma'}$, where $\sigma' \sim \mathcal U(\mathcal S_K^\varepsilon)$ with a weight proportional to $\lvert \mathcal S_K^\varepsilon \rvert$.

This methodology produces samples from the classical ABC pseudo-posterior distribution:
\begin{equation}
    \begin{split}
        \bar \pi_{\varepsilon} (\param \mid \obsdata) & \propto \pi(\param) \int_{A_{\obsdata,\varepsilon}} \frac{f(\simdata \mid \param)}{\lvert \mathcal S_K^\varepsilon \rvert} d\simdata \cdot \lvert \mathcal S_K^\varepsilon \rvert\\
        & \propto \pi_{\varepsilon} (\param \mid \obsdata)
    \end{split}
\end{equation}

However, this ABC scheme is impractical because it is too costly. Sampling $K!$ permutations is unfeasible. Thus, we can write this last quantity $\lvert \mathcal S_K^\varepsilon \rvert$ as $\mathbb{E}_{\sigma \sim \mathcal{U}(\mathcal{S}_K)}\left[\mathbb{I}_{A_{\obsdata,\varepsilon}}(\simdata_{\sigma})\right]$, which we will estimate by Monte Carlo. For small $\varepsilon$, a standard Monte Carlo simulation will have large variance: indeed, for most permutations $\sigma$, $d(\eta(\obsdata), \eta(\simdata_{\sigma})) > \varepsilon$, i.e., $\mathbb{I}_{A_{\obsdata,\varepsilon}}(\simdata_{\sigma})=0$. 

Instead, we resort to importance sampling. We write

\[\mathbb{E}_{\sigma \sim \mathcal{U}(\mathcal{S}_K)}\left[\mathbb{I}_{A_{\obsdata,\varepsilon}}(\simdata_{\sigma})\right] = 
\mathbb{E}_{\sigma \sim \rho}\left[\frac{\mathbb{I}_{A_{\obsdata,\varepsilon}}(\simdata_{\sigma})}{K!\rho(\sigma)}\right]\]

To do this, we sample $\sigma_{\iperm} \sim \rho(\sigma)$ for $\iperm=1,\dots, L$ and replace the weights $\lvert \mathcal S_K^\varepsilon \rvert$ by the following importance sampling estimator: $\frac{1}{L}\sum_{l=1}^L \frac{\mathbb{I}_{A_{\obsdata,\varepsilon}}(\simdata_{\sigma_{\iperm}})}{\rho(\sigma_{\iperm})}$. 

We now look for a good proposal distribution $\rho$ over $\mathcal{S}_K$.

\subsection{Simulating Non-uniform Permutations}

A good proposal $\rho$ will have support all of the symmetric group $\mathcal{S}_K$, but will be concentrated on values which are more likely to verify $\sigma$, $d(\eta(\obsdata), \eta(\simdata_{\sigma})) < \varepsilon$. We proceed in two steps:

\begin{itemize}
    \item We compute $\sigma^* = \arg\min_{\sigma} d(\eta(\obsdata), \eta(\simdata_{\sigma}))$, which can be done efficiently using a Jonker-Volgenant algorithm. (Complexity $\mathcal{O}(K^3)$)
    \item If $d(\eta(\obsdata), \eta(\simdata_{\sigma^*})) > \varepsilon$, we must reject $\param$. On the other hand, if $d(\eta(\obsdata), \eta(\simdata_{\sigma^*})) < \varepsilon$, we use a distribution $\rho(\cdot\mid\sigma^*)$ concentrated near $\sigma^*$.
\end{itemize}

We now propose a strategy to devise such a distribution $\rho(\cdot\mid\sigma^*)$. An additional constraint is that since this distribution will be used in importance sampling, its probability mass function must be computable pointwise. Our distribution is best understood through its generative process:

\begin{enumerate}
    \item Generate an auxiliary variable $N$ from a distribution $p_N$ with support $\{1, \ldots, K\}$; we choose the arbitrary distribution $N \sim 1 + \text{Poisson}(\lambda)$ truncated at $K$.
    \item If $N = 1$ then return $\sigma = \sigma^*$.
    \item If $N > 1$, generate $\sigma$ uniformly from the set $\mathfrak{S}_N = \left\{\sigma \in \mathcal{S}_K: \lVert \sigma - \sigma^*\rVert_0 = N\right\}$ where $\lVert \cdot \rVert_0$ is the Hamming or edit distance: $\lVert\sigma - \sigma^*\rVert_0 = \sum_{j=1}^K \mathbb{I}_{\{\sigma(i) \neq \sigma^*(i)\}}$.
\end{enumerate}

Note that this choice is valid because $\cup_n \mathfrak{S}_n = \mathcal{S}_K$ and because $\mathfrak{S}_1 = \emptyset$ and $\mathfrak{S}_0 = \{\sigma^*\}$. 

To generate uniformly from $\mathfrak S_n$, we select uniformly $n$ elements out of $K$, generate a random derangement of length $n$, and apply it to the selected elements. Recall that a derangement is a permutation that leaves no point invariant; the number of derangements of length $n$ is called the subfactorial of $n$, denoted by $!n$ and is easily computed: $!n = n! \sum_{k=0}^n \frac{(-1)^k}{k!} = \lceil \frac{n!}{e} \rfloor$ where $\lceil \cdot \rfloor$ denotes the closest integer.
The set $\mathfrak{S}_n$ is known as the set of the partial derangements and is then such that $\lvert \mathfrak S_n \rvert = \binom{K}{n} !n$. 
Thus the probability of sample $\sigma^*$ is: $\rho(\sigma^*\mid \sigma^*) = p_N(1)$ and for  permutation $\sigma \neq \sigma^*$, let $n = \lVert\sigma - \sigma^*\rVert_0$,  the mass function at $\sigma$ is:

$\rho(\sigma\mid\sigma^*) = p_N(n)\frac{1}{\lvert \mathfrak{S}_n \rvert} = p_N(n)\frac{1}{\binom{K}{n}!n}$

\subsection{Stratified estimator}

With this instrumental distribution centered at $\sigma^*$, we can enhance the effectiveness of our method through the implementation of a stratified estimator. In doing so, we can manage the permutations we randomly sample, ensuring acceptance of parameter $\parprop$ if at least one permutation (specifically $\sigma^*$) satisfies the ABC conditions.

We establish a partition of the set $\mathcal S_K$ into $H$ strata denoted as $(D_h)$ with corresponding weights $(w_h)$ for $h=1,\dots,H$. Here, $D_1 = \mathfrak S_0 = \{ \sigma^* \}$, $D_h = \mathfrak S_h$ for $h = 2,\dots,H-1$, and finally $D_H = \cup_{h=H}^K \mathfrak S_h$, with $w_h = p_N(h)$ for $h=1,\dots,H-1$ and $w_H = \sum_{h'=H}^K p_N(h')$. Subsequently, we define a new sampling distribution as a mixture across our strata: $\rho(\sigma)= \sum_{h=1}^H w_h \cdot \rho_h(\sigma)$, where $\rho_h$ represents the distribution within each stratum $D_h$: $\rho_h = \mathcal U(D_h)$ for $h=1,\dots,H-1$, and $\rho_H = \frac{1}{w_H}\sum_{h=H}^K p_N(D_h)\cdot \mathcal U(\mathfrak S_h)$. Additionally, we select the positive number $\Npermstrat{h}>0$ of permutations within each stratum. Thus, for the initial stratum, we must have $\Npermstrat{1} = 1$, ensuring consideration of $\sigma^*$, and $\Npermstrat{h}>0$ for $h=2,\dots,H$. Consequently, we derive the following stratified estimator with $\sigma_{\iperm}^{(h)} \sim \rho_h$ for $h=1,\dots,H$ and $\iperm=1,\dots,\Npermstrat{h}$:

$\mathbb E_{\sigma \sim \mathcal U(\mathcal S_K)} \left [ \mathbb I_{A_{\obsdata,\varepsilon}}(\simdata_{\sigma}) \right] \approx \sum_{h=1}^H \frac{1}{\Npermstrat{h}} \sum_{l=1}^{\Npermstrat{h}}\frac { \mathbb I_{A_{\obsdata,\varepsilon}}(\simdata_{\sigma_l^{(h)}})}{\rho_h(\sigma_l^{(h)})} $

In this section, we have presented an effective ABC simulation method that seamlessly integrates with the traditional accept-reject ABC Algorithm, often referred to as ABC Vanilla. This integration leads to the formulation of Algorithm \ref{alg:permABC_star}, which we specifically label as permABC Vanilla.

\begin{algorithm}[ht]
\caption{Stratified permABC Vanilla}
\textbf{Input: } observed dataset $\obsdata$, number of iterations $N$, threshold $\varepsilon > 0$, number of strata $H$, proposal distributions $\rho_1, \dots, \rho_H$. \\
\textbf{Output: } a sample $\{\param^{(1)}, \dots, \param^{(N)}\}$ from the permuted ABC approximation $\pi_\varepsilon(\cdot \mid \obsdata)$.\\

\BlankLine

\For{$i = 1, \dots, N$}{
    \Repeat{$d(\obsdata, \simdata_{\sigma^*}) \leq \varepsilon$}{
        $\parprop \sim \pi(\cdot)$\;
        $\simdata \sim f(\cdot \mid \parprop)$\;
        $\sigma^* = \argmin_{\sigma} d(\obsdata, \simdata_\sigma)$\;
        $d_i^* = d(\obsdata, \simdata_{\sigma^*})$\;
    }

    \For{$h = 1, \dots, H$}{
        Generate permutations $\sigma_1^{(h)}, \dots, \sigma_{L_h}^{(h)} \sim \rho_h(\cdot)$\;
        \For{$l = 1, \dots, L_h$}{
            Compute $d_l^{(h)} = d(\obsdata, \simdata_{\sigma_l^{(h)}})$\;
            Compute weight $\omega_l^{(h)} = \frac{\mathbb{I}\{d_l^{(h)} < \varepsilon\}}{\rho_h(\sigma_l^{(h)})}$\;
        }
    }

    Normalize all weights: $\tilde{\omega}_l^{(h)} = \omega_l^{(h)} / \sum_{h=1}^H \sum_{l=1}^{L_h} \omega_l^{(h)}$\;

    Sample $\sigma^{\text{prop}}$ among $\{\sigma_l^{(h)}\}_{l,h}$ with probability $\tilde{\omega}_l^{(h)}$\;

    Set $\param^{(i)} = \parprop_{\sigma^{\text{prop}}}$\;

    Assign weight $W_i \propto \sum_{h=1}^H \frac{1}{L_h} \sum_{l=1}^{L_h} \frac{\mathbb{I}\{d(\obsdata, \simdata_{\sigma_l^{(h)}}) < \varepsilon\}}{\rho_h(\sigma_l^{(h)})}$\;
}
\end{algorithm}

\section{Theoretical Guarantees for permABC}
\label{appendix_proofs}

We now provide the proofs of the theoretical results stated in Section~\ref{sec:comparison}. The key insight is that for small enough $\varepsilon$, the optimal permutation is unique, which implies that permABC and standard ABC accept the same samples.

\begin{proof}[Proof of Lemma~\ref{lempermABC_star}]
Let $\sigma^* := \arg\min_{\sigma \in \mathcal{S}_K} d(\obsdata, \simdata_{\sigma})$ and assume $\varepsilon < \varepsilon^*$. We consider two cases:

\textbf{Case 1:} If $d(\obsdata, \simdata_{\sigma^*}) > \varepsilon$, then no permutation satisfies the ABC condition:
\[
\mathrm{Card} \left( \left\{ \sigma \in \mathcal{S}_K : d(\obsdata, \simdata_{\sigma}) \leq \varepsilon \right\} \right) = 0.
\]

\textbf{Case 2:} Suppose without loss of generality that $\sigma^* = \mathrm{Id}$ and $d(\obsdata, \simdata) \leq \varepsilon$. For any $\sigma \neq \mathrm{Id}$, we have:
\[
\begin{aligned}
d(\obsdata, \simdata_{\sigma}) 
&\geq \left| d(\obsdata, \obsdata_{\sigma}) - d(\obsdata_{\sigma}, \simdata_{\sigma}) \right| \\
&= \left| d(\obsdata, \obsdata_{\sigma}) - d(\obsdata, \simdata) \right| \\
&> \varepsilon,
\end{aligned}
\]
since $d(\obsdata, \obsdata_{\sigma}) > 2\varepsilon$ and $d(\obsdata, \simdata) \leq \varepsilon$. Therefore, only the identity permutation satisfies the ABC criterion.
\end{proof}

\begin{proof}[Proof of Theorem~\ref{th:permABC_star}]
Let $\varepsilon < \varepsilon^*$. For every simulated dataset $\simdata \in B_\varepsilon(\obsdata)$ such that $d(\obsdata, \simdata) \leq \varepsilon$, Lemma~\ref{lempermABC_star} implies that at most one permutation satisfies $d(\obsdata, \simdata_{\sigma}) \leq \varepsilon$. In that case, the optimal permutation is necessarily $\sigma = \mathrm{Id}$.

Consequently, permABC and standard ABC both accept the same set of simulations, and:
\[
A^*_{\obsdata,\varepsilon} = B_{\varepsilon}(\obsdata) \cap \dataspace_{\obsdata}^{*K} = B_{\varepsilon}(\obsdata) = A_{\obsdata,\varepsilon},
\]
which implies:
\[
\pi^*_{\varepsilon}(\cdot \mid \obsdata) = \pi_{\varepsilon}(\cdot \mid \obsdata).
\]
\end{proof}

\begin{proof}[Proof of Corollary~\ref{col:permabc}]
By Theorem~\ref{th:permABC_star}, we have $\pi^*_{\varepsilon}(\cdot \mid \obsdata) = \pi_{\varepsilon}(\cdot \mid \obsdata)$ for all $\varepsilon < \varepsilon^*$. If $\pi_{\varepsilon}(\cdot \mid \obsdata) \to \pi(\cdot \mid \obsdata)$ as $\varepsilon \to 0$, then the same convergence holds for $\pi^*_{\varepsilon}(\cdot \mid \obsdata)$.
\end{proof}

\section{Technical Details for Over Sampling}
\label{appendix_os}

\subsection{Formal Definition of \( \pi^{*M}_\varepsilon \)}

Given a dataset \( \obsdata \in \dataspace^K \), the over sampling acceptance condition selects particles whose associated simulated data \( \simdata \in \dataspace^M \) satisfy the constraint:
\[
\min_{\sigma \in \mathcal A_K^M} d(\obsdata, \simdata_\sigma) \leq \varepsilon.
\]
We denote the acceptance region by
\[
\tilde A_{\varepsilon, \obsdata}^{M} = \left\{ \simdata \in \dataspace^M \mid \min_{\sigma \in \mathcal A_K^M} d(\obsdata, \simdata_\sigma) \leq \varepsilon \right\}.
\]

The associated joint pseudo-posterior is then defined as:
\[
\tilde{\pi}_\varepsilon^M(\param, \simdata \mid \obsdata) \propto \pi(\param) f(\simdata \mid \param) \, \mathbb{I}_{\tilde A_{\varepsilon, \obsdata}^{M}}(\simdata).
\]

To restore identifiability, we apply a projection \( \proj \) that reorders the parameters and data. Specifically, the optimal partial permutation \( \sigma^* \in \mathcal A_K^M \) is used to assign the best-matching \( K \) simulated compartments to the observed data. The remaining \( M - K \) components are sorted by index. We define:
\[
\proj(\param, \simdata) := (\param_{\sigma^*}, \simdata_{\sigma^*}),
\]
and the projected joint pseudo-posterior becomes:
\[
\pi^{*M}_{\varepsilon}(\param, \simdata \mid \obsdata)
:= \projsharp \tilde{\pi}_{\varepsilon}^{M}(\param, \simdata \mid \obsdata)
\]

This distribution lives on \( \Theta^M \times \dataspace^M \), but we are primarily interested in the marginals over the global parameter \( \globparam \) and the local parameters \( \locparam_{1:K} \) corresponding to the optimally matched compartments.

\subsection{Choosing \( \varepsilon \) Given \( M_0 \)}

As discussed in Section~\ref{sec:permABC_OS}, the acceptance condition \eqref{eq:os} becomes easier to satisfy as \( M \) increases. Consequently, for a given tolerance \( \varepsilon \), a larger value of \( M \) results in a higher acceptance probability. To ensure that a sufficient number of particles are alive in the initial iteration, one must jointly select \( M_0 \) and \( \varepsilon \) so as to reach a desired acceptance rate.

Empirically, the optimal assignment distance \( \min_{\sigma} d(\obsdata, \simdata_\sigma) \) tends to decrease monotonically with \( M \). Given a simulation budget of \( N \) particles, we define the calibration of \( \varepsilon \) as the smallest value such that at least a proportion \( p \) of the simulated datasets satisfy the condition:
\[
\varepsilon := \inf \left\{ \varepsilon > 0 \mid \frac{1}{N} \sum_{i=1}^N \mathbb{I} \left( \min_\sigma d(\obsdata, \simdata^{(i)}_\sigma) \leq \varepsilon \right) \geq p \right\}.
\]

In other words, \( \varepsilon \) is chosen as the empirical quantile of order \( p \) of the distances \( \min_\sigma d(\obsdata, \simdata^{(i)}_\sigma) \), computed on the initial sample. This guarantees that a proportion approximately $p$ of the particles satisfy the acceptance criterion and can thus be propagated into the first population. In practice, we recommend using \( p = 0.95 \).

\subsection{Designing the Sequence \( M_t \)}

To avoid sudden drops in the acceptance rate between iterations, the sequence \( M_t \) must decay smoothly. A practical heuristic is to choose:
\[
M_{t+1} = \lfloor K + (M_t - K) \gamma \rfloor \wedge M_t -1
\]
for some decay exponent \( \gamma > 0 \). This schedule allows for a slow reduction at the beginning (when population diversity is critical), followed by a sharper decay toward the final iteration where convergence to \( M_T = K \) is enforced. In practice, we recommend using a conservative value such as \( \gamma = 0.9 \).

\subsection{Particle Duplication Strategy}

Even under a conservative schedule such as \( M_{t+1} = M_t - 1 \), the probability of discarding particles during the transition from iteration \( t \) to \( t+1 \) can remain high. This is due to the potential exclusion of compartments that previously matched the observed data under the acceptance criterion \eqref{eq:os}.

More precisely, the probability that a matched compartment from iteration \( t \) is excluded from the optimization at iteration \( t+1 \) can be upper bounded by:
\[
1 - \frac{(M_t - K)!}{(M_{t+1} - K)!} \cdot \frac{M_{t+1}!}{M_t!},
\]
which equals \( K/(K+1) \) in the case \( M_t = K + 1 \) and \( M_{t+1} = K \). This bound highlights the non-negligible risk of mass extinction, especially in the final iterations.

To mitigate this, we introduce a duplication mechanism. Before evaluating the acceptance condition at iteration \( t+1 \), each accepted particle from iteration \( t \) is replicated \( R \) times using independent random partial permutations \( \sigma_1, \dots, \sigma_R \in \mathcal{A}_K^{M_{t+1}} \), where we may assume \( \sigma_1 = \mathrm{Id} \) for simplicity. For each duplicated version, we test the acceptance criterion on \( \simdata_{\sigma_r} \) for \( r = 1, \dots, R \).

This procedure increases the likelihood that at least one configuration of the particle will pass the acceptance test, without requiring additional simulations of the likelihood. However, it does incur additional computational cost due to the \( R \) assignment problems per particle. In practice, we find that values such as \( R = 5 \) or \( R = 10 \) are sufficient to ensure stability while keeping the cost reasonable. The value of \( R \) can be chosen adaptively based on the upper bound described above.

\section{Technical Details for Under Matching}
\label{appendix_um}

\subsection{Distance, Acceptance Region, and Projection}

As introduced in Section~\ref{sec:undermatching}, the under matching strategy accepts simulated datasets \( \simdata \in \dataspace^K \) when a subset of \( L < K \) compartments can be matched to observed data. The acceptance region is defined as:
\[
\tilde A_{\varepsilon, \obsdata}^{L} = \left\{ \simdata \in \dataspace^K \mid \min_{\sigma, \tilde{\sigma} \in \mathcal A_L^K} d(\obsdata_{\tilde \sigma}, \simdata_\sigma) \leq \varepsilon \right\}.
\]

The pseudo-posterior then takes the form:
\[
\tilde{\pi}_\varepsilon^L(\param, \simdata \mid \obsdata) \propto \pi(\param) f(\simdata \mid \param) \, \mathbb{I}_{\tilde A_{\varepsilon, \obsdata}^{L}}(\simdata).
\]

As in the over sampling case, we define a projection operator \( \proj \) that maps each accepted pair \( (\param, \simdata) \) to a reordered version where the \( L \) optimally matched components are moved to the top \( L \) positions, and the remaining \( K - L \) components follow in order. This projected distribution, denoted \( \pi^{*L}_\varepsilon \), focuses on inference for the global parameter and the best-matching local components.

\subsection{Designing the Sequence \( L_t \)}

To guide the algorithm toward the full matching regime, we define a growing schedule:
\[
L_{t+1} = \lfloor L_t + (K - L_t) \cdot (1-\gamma) \rfloor \vee L_t + 1
\]
for some smoothing exponent \( \gamma \in (0, 1] \). This ensures that the number of matched compartments increases gradually toward \( K \), with smaller changes in the early iterations and a final push to full matching.

In practice, values such as \( \gamma = 0.8 \) or \( 0.9 \) offer a good trade-off between gradual adaptation and final convergence. The stopping point of the algorithm is \( L_T = K \), at which point the standard permABC acceptance condition is recovered.

\section{Blockwise Metropolis-within-Gibbs Kernel}
\label{appendix_kernel}

This section details the MCMC transition kernel used in the permABC-SMC algorithm. At each iteration, the algorithm performs both a global move and a local move, ensuring comprehensive exploration of the posterior while maintaining computational efficiency through partial updates.

The global move proposes a new value for the shared parameter \( \globparam \) using a symmetric random walk proposal kernel \( q_t \), as defined in Section~\ref{sec:kernels}. Synthetic data are then generated for all compartments using the proposed value \( \globparam^* \), and the optimal permutation minimizing the distance to the observed data is computed. The proposal is accepted or rejected based on the ABC Metropolis–Hastings criterion, using the full data discrepancy.

The local move updates the collection of local parameters \( \locparam_1, \dots, \locparam_K \) blockwise. The set of compartments is randomly partitioned into \( H \) equally sized blocks. Within each block, new parameters are proposed using the same Gaussian random walk kernel \( q_t \), and new synthetic data are generated only for the updated compartments. The new configuration is accepted or rejected based on the global ABC condition, ensuring coherence of the posterior sample across all compartments.

Performing both global and local moves at every iteration leads to a total of \( 2NK \) data simulations per iteration for a particle population of size \( N \), and requires \( N(1 + H) \) linear assignment resolutions: one for the global move, and \( H \) for the local updates.

Although this construction is reminiscent of ABC-Gibbs methods, such as \citet{clarte2021componentwise}, the key distinction lies in the evaluation of the ABC discrepancy. In our case, the acceptance criterion remains global—defined over the full dataset—regardless of whether global or local parameters are being updated. This ensures that accepted configurations correspond to coherent fits to the observed data across all compartments.

In the Over Sampling setting, we consider a total of \( M > K \) compartments. At iteration \( t \), the kernel operates on the current set of \( M_t \) simulated compartments, among which only \( K \) are matched to the observed data under the optimal injection. These \( K \) matched compartments are partitioned into \( H \) blocks for Gibbs updates using the kernel \( q_t \), while the remaining \( M_t - K \) unmatched compartments form a separate block, updated independently by redrawing from the prior. Components beyond \( M_t \) are frozen.

In the Under Matching setting, the kernel operates on the current subset of \( L_t < K \) matched compartments. These \( L_t \) compartments are partitioned into \( H \) blocks, each updated via random walk proposals \( q_t \), while the \( K - L_t \) unmatched compartments are handled separately and sampled directly from the prior. In both cases, only the matched components contribute to the ABC discrepancy and drive acceptance.

\begin{algorithm}[ht]
\caption{MCMC kernel with global and local updates for permABC-SMC}
\label{alg:kernel}

\textbf{Input:} Current state $(\param, \simdata, \sigma)$, tolerance $\varepsilon$, number of blocks $H$ \\
\textbf{Output:} New state $(\tilde \param, \tilde \simdata, \tilde \sigma)$ targeting $\pi_\varepsilon(\cdot, \cdot \mid \obsdata)$

\BlankLine

$\tilde \param \leftarrow \param$, $\tilde \simdata \leftarrow \simdata$, $\tilde \sigma \leftarrow \sigma$ \\

\vspace{0.5em}
\textit{Global update:} \\
$\globparam^* \sim q_t(\cdot \mid \globparam)$ \\
Generate: $\simdata_k^* \sim g(\cdot \mid \globparam^*, \locparam_k)$ for $k = 1, \dots, K$ \\
Compute: $\sigma^* \leftarrow \argmin_{\sigma \in \mathcal{S}_K} d(\simdata^*_\sigma, \obsdata)$ \\
Draw $U_g \sim \mathcal{U}(0,1)$ \\
\If{
$U_g \leq \dfrac{\pi(\globparam^*) q_t(\globparam \mid \globparam^*, \sigma^*)}{\pi(\globparam) q_t(\globparam^* \mid \globparam, \sigma)} \cdot \mathbb{I}_{\tilde A_{\obsdata, \varepsilon}}(\simdata^*)$
}{
    $\tilde \globparam \leftarrow \globparam^*$ \\
    $\tilde \simdata \leftarrow \simdata^*$\\
    $\tilde \sigma \leftarrow \sigma^*$ \\
}

\vspace{0.5em}
\textit{Local update:} \\
Randomly partition $\{1, \dots, K\}$ into $H$ blocks $(B_1, \dots, B_H)$ \\

\For{$h = 1, \dots, H$}{
    $\simdata^* \leftarrow \tilde \simdata$ \\
    \For{$k \in B_h$}{
        Propose: $\locparam_k^* \sim q_t(\cdot \mid \tilde \locparam_k, \tilde \sigma)$ \\
        Simulate: $\simdata_k^* \sim g(\cdot \mid \tilde \globparam, \locparam_k^*)$
    
        Compute: $\sigma^* \leftarrow \argmin_{\sigma \in \mathcal{S}_K} d(\simdata^*_\sigma, \obsdata)$ \\
        Draw $U_h \sim \mathcal{U}(0,1)$ \\
        \If{
        $U_h \leq \dfrac{\pi(\locparam_{B_h}^*) q_t(\tilde \locparam_{B_h} \mid \locparam_{B_h}^*, \sigma^*)}{\pi(\tilde \locparam_{B_h}) q_t(\locparam_{B_h}^* \mid \tilde \locparam_{B_h}, \tilde \sigma)} \cdot \mathbb{I}_{\tilde A_{\obsdata, \varepsilon}}(\simdata^*)$
        }{
            $\tilde \locparam_{B_h} \leftarrow \locparam_{B_h}^*$\\
            $\tilde \simdata_{B_h} \leftarrow \simdata^*_{B_h}$\\
            $\tilde \sigma \leftarrow \sigma^*$ \\
        }
    }
}

\textbf{Return:} $(\tilde \param, \tilde \simdata, \tilde \sigma)$ with $\tilde \param = (\tilde \globparam, \tilde \locparam)$

\end{algorithm}

\section{Validation and Results of the SIR Model Application}
\label{appendix_sir}

\subsection{Preliminary Justification for the Hierarchical Structure}
\label{appendix_sir_validation}

Before applying permABC to the full departmental dataset, we performed a preliminary analysis to justify the modeling assumptions. In particular, we fitted independent SIR models to each department using standard ABC techniques. These runs revealed that although local dynamics exhibit variability in terms of initial conditions and intensity, the inferred values of the reproduction number \( R_0 \) were broadly consistent across departments during the first wave of the epidemic (March–July 2020).

This empirical observation supports the use of a global \( R_0 \) parameter shared across all compartments, while allowing for department-specific local parameters \( \locparam_k = (I_k(0), R_k(0), \delta_k) \).

The spatial distribution of \( \delta_k \) reflects variations in local factors such as population density, urbanization, mobility, or healthcare capacity. These differences highlight the necessity of treating local parameters separately in a hierarchical framework, rather than aggregating data at the national level.

\end{document}